\pdfoutput=1

\documentclass[a4paper]{scrartcl}
\usepackage[utf8]{inputenc}
\usepackage[UKenglish]{babel}
\usepackage{amssymb}
\usepackage{amsmath}
\usepackage{amsthm}
\usepackage{enumerate}
\usepackage{listings}
\usepackage{mathtools}
\usepackage{microtype}
\usepackage{stmaryrd}
\usepackage{hyperref}
\usepackage[capitalise, noabbrev]{cleveref}
\usepackage[mathscr]{eucal}

\newcommand{\NN}{\mathbb{N}}
\newcommand{\NNpos}{\NN_{\scriptscriptstyle \geq 1}}
\newcommand{\ZZ}{\mathbb{Z}}
\newcommand{\QQ}{\mathbb{Q}}
\newcommand{\RR}{\mathbb{R}}

\newcommand{\Structure}[1]{\ensuremath{\mathcal{#1}}}
\newcommand{\A}{\Structure{A}}
\newcommand{\B}{\Structure{B}}
\newcommand{\C}{\Structure{C}}
\newcommand{\D}{\Structure{D}}
\newcommand{\I}{\Structure{I}}
\newcommand{\X}{\Structure{X}}

\newcommand{\bigO}{\mathcal{O}}

\newcommand{\bigOh}{\bigO}

\newcommand{\FO}{\ensuremath{\textup{FO}}}
\newcommand{\FOC}{\ensuremath{\textup{FOC}}}
\newcommand{\FOunC}{\ensuremath{\FOC_1}}
\newcommand{\FOCNP}{\ensuremath{\textup{FOCN}(\mathbb{P})}}

\newcommand{\FOunCP}{\ensuremath{\FOunC(\mathbb{P})}}

\newcommand{\neighb}[3]{\ensuremath{N_{#1}^{#2}(#3)}} % neighbourhood (set of elements)
\newcommand{\Neighb}[3]{\ensuremath{\mathcal{N}_{#1}^{#2}(#3)}} % sphere: Neighbourhood (structure)
\newcommand{\Neighbr}[2]{\ensuremath{\mathcal{N}_{r}^{#1}(#2)}} % sphere: Neighbourhood (structure)

\newcommand{\type}[3]{\ensuremath{\textup{tp}_{#1}(#2,#3)}} % type

\newcommand{\abs}[1]{\left\lvert#1\right\rvert}

\DeclareMathOperator*{\ar}{ar}

\DeclareMathOperator{\dist}{dist}
\DeclareMathOperator*{\free}{free}

\DeclareMathOperator*{\qr}{qr}

\renewcommand{\leq}{\leqslant}
\renewcommand{\geq}{\geqslant}
\renewcommand{\le}{\leq}
\renewcommand{\ge}{\geq}

\renewcommand{\phi}{\varphi}
\renewcommand{\epsilon}{\varepsilon}

\newtheorem{thm}{Theorem}[section]
\newtheorem{lem}[thm]{Lemma}
\newtheorem{cor}[thm]{Corollary}
\theoremstyle{definition}
\newtheorem{defn}[thm]{Definition}
\newtheorem{prop}[thm]{Property}
\newtheorem{problem}[thm]{Problem}
\newtheorem{exmp}[thm]{Example}
\newtheorem{rem}[thm]{Remark}
\newtheorem{claim}{Claim}

\newenvironment{claimproof}[1][\proofname]{%
  
  \begin{proof}[#1]%
}{%
  \end{proof}%
}

\makeatletter
\@addtoreset{claim}{thm}
\makeatother

\usepackage[ruled]{algorithm}
\usepackage[noend]{algpseudocode}
\algnewcommand{\algorithmicreject}{\textbf{reject}}
\algnewcommand{\Reject}{\State \algorithmicreject}
\algnewcommand\AlgAnd{\textbf{and} }
\algnewcommand\AlgOr{\textbf{or} }
\algnewcommand\True{\textbf{true} }
\algnewcommand\False{\textbf{false} }
\algnewcommand\algorithmicalgorithm{\textbf{Algorithm}}
\algnewcommand\Algorithm{\item[\algorithmicalgorithm]}%

\algrenewcommand\Return{\State \textbf{return} }

\newcommand{\nc}[1]{\newcommand{#1}}
\newcommand{\rnc}[1]{\renewcommand{#1}}

\newcommand{\deff}{:=}

\nc{\ov}[1]{\bar{#1}}

\nc{\openinterval}[2]{\ensuremath{\textup{rat}(#1,#2)}}
\newcommand{\myparagraph}[1]{\par \indent \textbf{\textsf{#1.}}}

\newcommand{\set}[1]{\ensuremath{\{#1\}}}
\newcommand{\setc}[2]{\ensuremath{\set{#1 : #2}}}
\nc{\bigset}[1]{\ensuremath{\big\{ #1 \big\}}}
\nc{\bigsetc}[2]{\bigset{#1 : #2}}
\nc{\setsize}[1]{\ensuremath{|#1|}}
\nc{\bigsetsize}[1]{\ensuremath{\big|#1\big|}}
\nc{\Setsize}[1]{\bigsetsize{#1}}

\nc{\emptyword}{\ensuremath{\varepsilon}}
\nc{\emptytuple}{\ensuremath{()}}

\newcommand{\CollectionOfRingsFont}[1]{\mathbb{#1}}
\newcommand{\SC}{\CollectionOfRingsFont{S}} % collection of semi-rings
\newcommand{\Weights}{\ensuremath{\mathbf{W}}}
\newcommand{\weight}{\ensuremath{\mathtt{w}}}
\newcommand{\weightone}{\ensuremath{\mathtt{one}}}
\newcommand{\wtype}{\ensuremath{\textup{type}}}
\newcommand{\pweight}{\ensuremath{p}} % W-product

\newcommand{\plusSR}[1]{\ensuremath{+_{#1}}}
\newcommand{\minusSR}[1]{\ensuremath{-_{#1}}}
\newcommand{\malSR}[1]{\ensuremath{{\cdot}_{#1}}}
\newcommand{\nullSR}[1]{\ensuremath{0_{#1}}}
\newcommand{\einsSR}[1]{\ensuremath{1_{#1}}}

\newcommand{\plus}{\plusSR{}}
\newcommand{\mal}{\malSR{}}
\newcommand{\minus}{\minusSR{}}

\newcommand{\plusS}{\plusSR{S}}
\newcommand{\minusS}{\minusSR{S}}
\newcommand{\malS}{\malSR{S}}
\newcommand{\nullS}{\nullSR{S}}
\newcommand{\einsS}{\einsSR{S}}

\nc{\LogicL}{\ensuremath{\textup{L}}}
\newcommand{\WA}{\ensuremath{\textup{FOWA}}} % Weight Aggregation Logic
\newcommand{\WAPSR}[2]{\ensuremath{\WA(#1)[#2]}}
\newcommand{\WAPS}{\WAPSR{\PP}{\sigma,\SC,\Weights}}
\nc{\RL}{\WA}
\nc{\RLPSR}[2]{\WAPSR{#1}{#2}}
\nc{\RLPS}{\WAPS}

\newcommand{\WAun}{\ensuremath{\textup{FOWA}_1}} % Weight Aggregation Logic
\newcommand{\WAunPSR}[2]{\ensuremath{\WAun(#1)[#2]}}
\newcommand{\WAunPS}{\WAunPSR{\PP}{\sigma,\SC,\Weights}}
\nc{\RLun}{\WAun}
\nc{\RLunPSR}[2]{\WAunPSR{#1}{#2}}
\nc{\RLunPS}{\WAunPS}

\newcommand{\FOW}{\ensuremath{\textup{FOW}_1}}
\newcommand{\FOWSR}[2]{\ensuremath{\FOW(#1)[#2]}}
\newcommand{\FOWS}{\FOWSR{\PP}{\sigma,\SC,\Weights}}
\nc{\FOG}{\FOW}
\nc{\FOGGR}[2]{\FOWSR{#1}{#2}}
\nc{\FOGG}{\FOWS}

\newcommand{\Vars}{\ensuremath{\textsf{\upshape vars}}}
\newcommand{\varsof}{\ensuremath{\textup{\upshape vars}}}

\newcommand{\PP}{\ensuremath{\mathbb{P}}}
\newcommand{\Pred}{\ensuremath{\mathsf{P}}}
 
 \newcommand{\Ps}{\PP}

\newcommand{\ptype}{\ensuremath{\textup{type}}}

\nc{\Ssum}[2]{\ensuremath{\sum {#1}.{#2}}}

\nc{\und}{\ensuremath{\wedge}}
\nc{\Und}{\ensuremath{\bigwedge}}
\nc{\oder}{\ensuremath{\vee}}
\nc{\Oder}{\ensuremath{\bigvee}}
\nc{\nicht}{\ensuremath{\neg}}
\nc{\impl}{\ensuremath{\to}}
\nc{\gdw}{\ensuremath{\leftrightarrow}}

\nc{\sem}[1]{\llbracket #1 \rrbracket}

\nc{\dsum}{\ensuremath{\oplus}}
\nc{\dunion}{\ensuremath{\sqcup}}

\rnc{\S}{\ensuremath{\mathcal{S}}}

\nc{\nbset}[2]{\ensuremath{N_{#1}^{#2}}}
\nc{\nrA}[1]{\ensuremath{\neighb{r}{\A}{#1}}} % r-neighbourhood in
\nc{\nrAStrich}[1]{\ensuremath{\neighb{r}{\B}{#1}}} % r-neighbourhood
\nc{\nrT}[1]{\ensuremath{\neighb{r}{\T}{#1}}} % r-neighbourhood in
\nc{\nRT}[1]{\ensuremath{\neighb{R}{\T}{#1}}} % R-neighbourhood

\nc{\NrA}[1]{\ensuremath{\Neighb{r}{\A}{#1}}} % r-neighbourhood in
\nc{\NRA}[1]{\ensuremath{\Neighb{R}{\A}{#1}}} % R-neighbourhood in
\nc{\NrB}[1]{\ensuremath{\Neighb{r}{\B}{#1}}} % r-neighbourhood in
\nc{\NRB}[1]{\ensuremath{\Neighb{R}{\B}{#1}}} % R-neighbourhood in
\nc{\NrDStrich}[1]{\ensuremath{\Neighb{r}{\B}{#1}}} % r-neighbourhood
\nc{\NrT}[1]{\ensuremath{\Neighb{r}{\T}{#1}}} % r-neighbourhood in
\nc{\Class}{\ensuremath{\mathcal{C}}}
\nc{\GraphclassH}{\ensuremath{\mathcal{H}}}
\nc{\GraphclassG}{\ensuremath{\mathcal{G}}}
\nc{\Graphclass}{\ensuremath{\mathcal{G}}}

\DeclareMathOperator{\countr}{d_{\textup{ag}}} % aggregation-depth
\nc{\agg}{\ensuremath{\text{ag}}}

\nc{\Count}[2]{\ensuremath{\sum {#1}.{#2}}}
\nc{\quant}[1]{\ensuremath{\textsf{\upshape #1}}}
\rnc{\P}{\quant{P}}

\nc{\inducedSubStr}[2]{\ensuremath{#1[#2]}} %induced substructure

\nc{\ComplexityClassFont}[1]{\ensuremath{\textsc{#1}}}
\nc{\ACzero}{\ensuremath{\ComplexityClasFont{AC}^0}}
\nc{\PTIME}{\ensuremath{\ComplexityClassFont{Ptime}}}
\nc{\NP}{\ensuremath{\ComplexityClassFont{NP}}}
\nc{\PSPACE}{\ensuremath{\ComplexityClassFont{Pspace}}}

\nc{\FPL}{\ensuremath{\ComplexityClassFont{fpl}}}
\nc{\FPT}{\ensuremath{\ComplexityClassFont{fpt}}}
\nc{\AWstern}{\ensuremath{\ComplexityClassFont{AW}[*]}}
\nc{\Wone}{\ensuremath{\ComplexityClassFont{W}[1]}}

\nc{\sharpP}{\ensuremath{\#\ComplexityClassFont{P}}}
\nc{\sharpWone}{\ensuremath{\#\ComplexityClassFont{W}[1]}}
\title{Learning Concepts Described by\\ Weight Aggregation Logic}

\author{Steffen van Bergerem\\
  \normalsize RWTH Aachen University
  \and
  Nicole Schweikardt\\
  \normalsize Humboldt-Universität zu Berlin
}

\date{}

\begin{document}

\maketitle

\begin{abstract}
We consider weighted structures, which extend ordinary
relational structures by assigning weights,
i.e.\ elements from a particular group or ring, to tuples present
in the structure.
We introduce an extension of first-order logic that allows to
aggregate weights of tuples, compare such aggregates, and use them to
build more complex formulas.
We provide locality properties of fragments of this logic including
Feferman-Vaught decompositions and a Gaifman normal form for a
fragment called $\FOW$, as well as a localisation theorem for a
larger fragment called $\WAun$.
This fragment can express concepts from
various machine learning scenarios.
Using the locality properties,
we show that concepts
definable in $\WAun$ over a weighted
background structure of at most polylogarithmic degree are
agnostically PAC-learnable
in polylogarithmic time after pseudo-linear time preprocessing.
\end{abstract}

\section{Introduction}\label{sec:intro}

In this paper, we study Boolean classification problems.
The elements that are to be classified come from a set \(\X\),
the \emph{instance space}.
A \emph{classifier} on \(\X\) is a function
\(c \colon \mathcal{X} \to \{0,1\}\).
Given a \emph{training sequence} \(T\) of labelled examples
\((x_i, b_i) \in \mathcal{X} \times \{0,1\}\),
we want to find a classifier, called a \emph{hypothesis},
that can be used to predict the label of elements
from \(\X\) not given in \(T\).
We consider the following well-known frameworks for this setting
from computational learning theory.

In Angluin's model of \emph{exact learning}~\cite{Angluin_ExactLearning},
the examples are assumed to be generated using an unknown classifier,
the \emph{target concept}, from a known \emph{concept class}.
The task is to find a hypothesis that is consistent with the training sequence \(T\),
i.e.\ a function \(h \colon \mathcal{X} \to \{0,1\}\) such that
\(h(x_i) = b_i\) for all \(i\).
In Haussler's model of \emph{agnostic probably approximately correct (PAC) learning}~\cite{Haussler},
a generalisation of Valiant's \emph{PAC learning} model~\cite{Valiant_PAC},
an (unknown) probability distribution \(\mathcal{D}\)
on \(\X \times \{0,1\}\) is assumed
and training examples are drawn independently from this distribution.
The goal is to find a hypothesis that generalises well,
i.e.\ one is interested in algorithms that
return with high probability a hypothesis with a small expected error
on new instances drawn from the same distribution.
For more background on PAC learning,
we refer to~\cite{KearnsVazirani,Shalev-Shwartz:2014:UML:2621980}.
We study learning problems in the framework that was introduced by Grohe and
Tur{\'{a}}n~\cite{GroheTuran_Learnability} and further studied
in~\cite{GrienenbergerRitzert_Trees,GroheLoedingRitzert_MSO,GroheRitzert_FO,vanBergerem}.
There, the instance space \(\X\) is a set of tuples from a background structure
and classifiers are described using parametric models based on
logics.

\myparagraph{Our contribution}
We introduce a new logic for describing such classifiers, namely
\emph{first-order logic with weight aggregation} ($\WA$).
It operates on \emph{weighted structures}, which extend
ordinary relational structures by assigning weights, i.e.\ elements
from a particular abelian group or ring, to tuples present in the structure.
Such weighted structures were recently considered by Toru\'{n}czyk~\cite{Szymon},
who studied the complexity of query evaluation problems
for the related logic $\FO[\mathbb{C}]$ and its fragment
$\FO_{\textup{G}}[\mathbb{C}]$.
Our logic $\WA$, however, is closer to
the syntax and semantics of the first-order logic with counting
quantifiers $\FOC$ considered in~\cite{KuskeSchweikardt_FOC}.
This connection enables
us to achieve locality results for the fragments $\FOW$ and
$\WAun$ of $\WA$ similar to those obtained
in~\cite{KuskeSchweikardt_Gaifman,GroheSchweikardt_FOunC}.
Specifically, we achieve
Feferman-Vaught decompositions and a Gaifman normal form
for $\FOW$ as well as a localisation theorem for the
more expressive logic $\WAun$.
We provide examples illustrating that $\WAun$ can express concepts
relevant for various machine learning scenarios.
Using the locality properties,
we show that concepts definable in $\WAun$ over a weighted
background structure of at most polylogarithmic degree
are agnostically PAC-learnable
in polylogarithmic time after pseudo-linear time preprocessing.
This generalises the results that
Grohe and Ritzert~\cite{GroheRitzert_FO}
obtained for first-order logic
to the
substantially more expressive logic $\WAun$.

The main drawback of the existing logic-based learning results is
that they deal with structures and logics that are too weak for describing
meaningful classifiers for real-world
machine learning problems.
In machine learning, input data is often given via numerical values
which are contained in or extracted from a more complex structure,
such as a relational database (cf., \cite{DBLP:conf/pods/Grohe20,GroverLeskovec_node2vec,PanDing,DBLP:conf/sum/SchleichOK0N19}).
Hence, to combine these two types of information, we are interested in
hybrid structures, which extend relational ones by numerical values.
Just as in commonly used relational database systems,
to utilise the power of such hybrid structures,
the classifiers should be allowed to use different methods
to aggregate the numerical values.
Our main contribution is the design of a logic that is capable of
expressing meaningful machine learning problems and, at the same time,
well-behaved enough to have similar locality properties as first-order
logic, which enable us to learn the concepts in sublinear time.

\myparagraph{Outline}
This paper is structured as follows.
Section~\ref{sec:Prelim} fixes basic notation.
Section~\ref{sec:WA} introduces the logic $\WA$ and its fragments
$\FOW$ and $\WAun$,
provides examples, and discusses enrichments of the logic with
syntactic sugar in order to make it more user-friendly (i.e.\
easier to parse or construct formulas) without increasing its expressive
power.
Section~\ref{sec:Locality} provides locality results for the fragments
$\FOW$ and $\WAun$ that are similar in spirit to the known locality
results for first-order logic and the counting logic $\FOunC$.
Section~\ref{sec:PacLearning} is devoted to our results on agnostic
PAC learning.
Section~\ref{sec:conclusion} combines the results from the previous sections
to obtain our main learning theorem for $\WAun$, and concludes the paper with an
application scenario and directions for future work.

\section{Preliminaries}\label{sec:Prelim}

\myparagraph{Standard notation}
We write $\RR$, $\QQ$,
$\ZZ$, $\NN$, and $\NNpos$ for the sets of reals, rationals, integers, non-negative integers, and
positive integers, respectively.
For all $m,n\in\NN$, we write $[m,n]$ for the set
$\setc{k\in\NN}{m\le k\le n}$, and we let $[m]\coloneqq[1,m]$.
For a $k$-tuple $\ov{x}=(x_1,\ldots,x_k)$, we write $|\ov{x}|$ to
denote its \emph{arity} $k$.
By $\emptytuple$, we denote the empty tuple, i.e.\ the tuple of arity
0.
All graphs are assumed to be undirected.
For a graph $G$, we write $V(G)$ and $E(G)$ to denote its vertex set
and edge set, respectively. For $V'\subseteq V(G)$, we write $\inducedSubStr{G}{V'}$ to
denote the subgraph of $G$ induced on $V'$.

\myparagraph{Monoids, groups, semirings, and rings}
Recall that a \emph{monoid} is a set $M$ that is equipped
with a binary operator $\circ \colon M\times M\to M$ that is associative
and has a neutral element $e_M\in M$
(i.e.\ for all $a,b,c\in M$, we have $a\circ (b\circ c) = (a\circ b)\circ c$ and $e_M\circ a
=a\circ e_M = a$).
A monoid is
\emph{commutative} if $a\circ b=b\circ a$ holds for all $a,b\in M$.
An \emph{abelian group} is a commutative monoid $(M,\circ)$ where for
each $a\in M$, there is an $a'\in M$ such that $a\circ a'=e_M$; by
convention, we write ${-}a$ for this $a'$.
When referring to an abelian group, we usually write $(S,\plusS)$
instead of $(M,\circ)$, and we denote the neutral element
by $\nullS$.
A \emph{semiring} is a set $S$ that is equipped with two binary operators $\plus$ and $\mal$ such that $(S,\plus)$ is a commutative monoid with a neutral element $\nullS \in S$,
$(S,\mal)$ is a monoid with a neutral element $\einsS\in S$,
multiplication distributes over addition,
and multiplication by $\nullS$ annihilates $S$
(i.e.\ for all $a,b,c\in S$, we have
$a \mal (b\plus c) = (a \mal b) \plus (a \mal c)$,
$(b\plus c) \mal a = (b\mal a) \plus (c\mal a)$,
$\nullS \mal a = a \mal \nullS = \nullS$).
A \emph{ring} is a semiring $(S,\plus,\mal)$ where $(S,\plus)$ is an
abelian group.
A semiring or ring $(S,\plus,\mal)$ is called \emph{commutative} if
the monoid $(S,\mal)$ is commutative.
In the following, we briefly write $S$ instead of
$(S,\plus,\mal)$, and we write $\plusS$, $\malS$, $\nullS$,
$\einsS$ to denote the operators $\plus$ and $\mal$ and their neutral
elements.

\myparagraph{Signatures, structures, and neighbourhoods}
A \emph{signature} $\sigma$ is a finite set of relation symbols.
Associated with every $R\in\sigma$ is an arity $\ar(R) \in \NN$.
A \emph{$\sigma$-structure} $\A$ consists of a finite
non-empty set $A$ called the \emph{universe} of $\A$ (sometimes denoted $U(\A)$),
and for each $R \in \sigma$ a relation
$R^{\A} \subseteq A^{\ar(R)}$.
The \emph{size} of $\A$ is $\abs{\A} \deff \abs{A}$.
Note that, according to these definitions, all considered
signatures and structures are \emph{finite}, signatures are
\emph{relational} (i.e.\ they do not contain any constants or function
symbols), and may contain relation symbols of arity 0 (the only two 0-ary
relations over a set $A$ are $\emptyset$ and
$\set{\emptytuple}$).

Let $\sigma'$ be a signature with $\sigma'\supseteq\sigma$.
A \emph{$\sigma'$-expansion} of a $\sigma$-structure $\A$ is a
$\sigma'$-structure $\B$ with universe $B$
such that $B=A$ and $R^{\B}=R^{\A}$ for every $R\in\sigma$.
If $\B$ is a $\sigma'$-expansion of
$\A$, then $\A$ is called the
\emph{$\sigma$-reduct} of $\B$.
A \emph{substructure} of a $\sigma$-structure $\A$ is a
$\sigma$-structure $\B$ with a universe $B\subseteq A$ and
$R^{\B}\subseteq R^\A$ for all $R\in\sigma$.
For a $\sigma$-structure $\A$ and a non-empty set
$B\subseteq A$, we write $\inducedSubStr{\A}{B}$ to denote the
\emph{induced substructure} of $\A$ on $B$,
i.e.\ the $\sigma$-structure with universe $B$ and
$R^{\inducedSubStr{\A}{B}} = R^\A \cap B^{\ar(R)}$ for every
$R\in\sigma$.

The \emph{Gaifman graph} $G_\A$ of a
$\sigma$-structure $\A$ is the graph with vertex set $A$
and an edge between two distinct vertices $a,b\in A$ iff there exists
$R\in\sigma$ and a tuple $(a_1,\ldots,a_{\ar(R)})\in R^{\A}$ such
that $a,b\in\set{a_1,\ldots,a_{\ar(R)}}$.  The structure $\A$ is
\emph{connected} if $G_{\A}$ is connected;
the \emph{connected components} of $\A$ are the connected components
of $G_{\A}$.
The \emph{degree} of $\A$ is the degree of $G_\A$,
i.e.\ the maximum number of neighbours of a vertex of $G_\A$.
The \emph{distance} $\dist^\A(a,b)$
between two elements $a,b\in A$ is the minimal number of edges of a path from
$a$ to $b$ in $G_\A$; if no such path
exists, we set $\dist^\A(a,b)\deff\infty$.
For a tuple $\ov{a}=(a_1,\ldots,a_k)\in A^k$ and an element $b\in A$, we let
$\dist^{\A}(\ov{a},b)\deff \min_{i\in[k]}\dist(a_i,b)$, and for a
tuple $\ov{b}=(b_1,\ldots,b_\ell)$, we let $\dist(\ov{a},\ov{b})\deff\min_{j\in[\ell]}\dist(\ov{a},b_j)$.

For every $r \ge 0$,
the \emph{$r$-ball of $\ov{a}$ in $\A$} is the set
$ N_r^\A(\ov{a}) \ = \ \setc{b\in A\,}{\,\dist^\A(\ov{a},b)\le r}$.
The \emph{$r$-neighbourhood of $\ov{a}$ in $\A$} is the structure
$\Neighb{r}{\A}{\ov{a}}  \deff  \inducedSubStr{\A}{\neighb{r}{\A}{\ov{a}}}$\,.

\section{Weight Aggregation Logic}\label{sec:WA}

This section introduces our new logic, which we
call \emph{first-order logic with weight aggregation}.
It is inspired by
the counting logic $\FOC$ and its fragment $\FOunC$, as
introduced in~\cite{KuskeSchweikardt_FOC, GroheSchweikardt_FOunC}, as
well as
the logic $\FO[\mathbb{C}]$ and its fragment
$\FO_{\textup{G}}[\mathbb{C}]$, which were recently introduced by
Toru\'{n}czyk in~\cite{Szymon}.
Similarly as in~\cite{Szymon}, we consider \emph{weighted
  structures}, which extend ordinary relational structures by assigning
a weight, i.e.\ an element of a particular
group or ring,
to tuples present
in the structure.
The syntax and semantics of our logic, however, are closer in spirit to
the syntax and semantics of the logic $\FOunC$, since this will enable
us to achieve locality results similar to those obtained
in~\cite{KuskeSchweikardt_Gaifman,GroheSchweikardt_FOunC}.

\myparagraph{Weighted structures}
Let $\sigma$ be a
signature.
Let $\SC$ be a collection of
rings and/or abelian groups.
Let $\Weights$ be a finite set of \emph{weight symbols}, such that
each $\weight\in\Weights$ has an associated \emph{arity}
$\ar(\weight)\in\NNpos$
and a \emph{type} $\wtype(\weight)\in\SC$.
A \emph{$(\sigma,\Weights)$-structure} is a $\sigma$-structure $\A$ that is enriched, for every $\weight\in \Weights$, by an interpretation
$\weight^\A \colon A^{\ar(\weight)} \to \wtype(\weight)$, which satisfies the following \emph{locality condition}:
if $\weight^\A(a_1,\ldots,a_k)\neq \nullS$ for
$S\deff\wtype(\weight)$, $k\deff \ar(\weight)$ and
$(a_1,\ldots,a_k)\in A^k$, then $k= 1$ or $a_1=\cdots = a_k$ or there
exists an $R\in \sigma$ and a tuple $(b_1,\ldots,b_{\ar(R)})\in R^\A$
such that $\set{a_1,\ldots,a_k}\subseteq
\set{b_1,\ldots,b_{\ar(R)}}$.
All notions that were introduced in Section~\ref{sec:Prelim} for
$\sigma$-structures carry over to $(\sigma,\Weights)$-structures in
the obvious way.

We will use the following as running examples throughout this section.
\begin{exmp}\label{example:WStr}
\begin{enumerate}[(a)]
\item\label{exmp:marketplaceStart}
  Consider an online marketplace that allows retailers
  to sell their products to consumers.
  The database of the marketplace contains a table with transactions,
  and each entry consists of an identifier, a customer, a product,
  a retailer, the price per item, and the number of items sold.
  We can describe the database of the marketplace as a weighted structure
  as follows.
  Let $(\QQ, +, \mal)$ be the field of rationals,
  let
  $\Weights$ contain two unary weight symbols $\mathtt{price}$ and
  $\mathtt{quantity}$ of type $(\QQ,+,\mal)$,
  let \(\sigma = \{T\}\),
  and let \(\A\) be a \((\sigma, \Weights)\)-structure
  such that the universe \(A\) contains
  the identifiers for the transactions, customers, products, and retailers.
  For every transaction, let \(T^\A\) contain the
  4-tuple \((i, c, p, r)\) consisting of the identifier for the transaction,
  the customer, the product, and the retailer.
  For every transaction identifier \(i\), let \(\mathtt{price}^\A(i)\)
  be the price per item in the transaction
  and \(\mathtt{quantity}^\A(i)\) be the number of items sold.
\item\label{exmp:social-networkStart}
  In a recent survey \cite{PanDing},
  Pan and Ding describe different approaches to
  represent social media users via embeddings into
  a low-dimensional vector space,
  where the embeddings are based on the users'
  social media posts\footnote{Among other applications,
    such embeddings might be used
  to predict a user's personality or political leaning.}.
  We represent the available data by a weighted structure $\A$ as follows.
  Consider the group \((\mathbb{R}^k, +)\), where
  \(\mathbb{R}^k\) is the set of \(k\)-dimensional real vectors
  and \(+\) is the usual vector addition, and let
  $\Weights$ contain a unary weight symbol $\mathtt{embedding}$ of
  type $(\RR^k,+)$.
  Let  \(\sigma = \{F\}\) and
  let \(\A\) be a \((\sigma, \Weights)\)-structure
  such that the universe \(A\) consists of the users of a social network.
  Let \(F^\A\) contain all pairs of users \((a, b)\)
  such that \(a\) is a follower of \(b\).
  For every user \(a \in A\), let \(\mathtt{embedding}^\A(a)\)
  be a \(k\)-dimensional vector representing \(a\)'s social media
  posts.
\item\label{exmp:edgeweightsStart}
Consider vertex-coloured edge-weighted graphs,
where
$R,B,G$ are unary relations of red, blue, and green vertices,
\(E\) is a binary relation of edges, and where every edge $(a,b)$ has an
associated \emph{weight} that is a $k$-dimensional vector of reals
(for some fixed number $k$). Such graphs can be viewed as
\((\sigma, \Weights)\)-structures $\A$, where
\(\sigma = \{E, R, B, G\}\),
$\Weights$ contains a binary weight symbol $\weight$ of type
$(\mathbb{R}^k,+)$ and
\(\weight^\A(a,b) \in \mathbb{R}^k\) for all edges $(a,b)\in E^\A$.
\end{enumerate}
\end{exmp}

Fix a countably infinite set $\Vars$ of \emph{variables}.
A \emph{$(\sigma,\Weights)$-interpretation} $\I=(\A,\beta)$ consists of a $(\sigma,\Weights)$-structure $\A$ and an \emph{assignment} $\beta \colon \Vars\to A$.
For $k\in\NNpos$, elements $a_1,\ldots,a_k\in A$, and $k$ distinct variables $y_1,\ldots,y_k$, we write
$\I\frac{a_1,\ldots,a_k}{y_1,\ldots,y_k}$ for the interpretation $(\A,\beta\frac{a_1,\ldots,a_k}{y_1,\ldots,y_k})$, where
$\beta\frac{a_1,\ldots,a_k}{y_1,\ldots,y_k}$ is the assignment $\beta'$ with $\beta'(y_i)=a_i$ for every $i\in [k]$ and $\beta'(z)=\beta(z)$ for all $z\in\Vars\setminus\set{y_1,\ldots,y_k}$.

\bigskip

\myparagraph{The weight aggregation logic $\WA$ and its restrictions
  $\WAun$ and $\FOW$}
Let $\sigma$ be a
signature, $\SC$ a
collection of
rings and/or abelian groups, and $\Weights$ a
finite set of weight symbols.
An \emph{$\SC$-predicate collection} is a 4-tuple
$(\PP,\ar,\ptype,\sem{\cdot})$ where $\PP$ is a countable set of
\emph{predicate names} and, to each $\Pred\in\PP$,
$\ar$ assigns an \emph{arity} $\ar(\Pred)\in\NNpos$,
$\ptype$ assigns a \emph{type} $\ptype(\Pred)\in \SC^{\ar(\Pred)}$,
and $\sem{\cdot}$ assigns a \emph{semantics}
$\sem{\Pred}\subseteq \ptype(\Pred)$.
For the remainder of this section, fix an $\SC$-predicate
collection $(\PP,\ar,\ptype,\sem{\cdot})$.

For every $S\in\SC$ that is not a ring but just an abelian group, a
\emph{$\Weights$-product of type $S$} is
either an element $s\in S$ or an expression of
the form $\weight(y_1,\ldots,y_k)$ where $\weight\in\Weights$ is of
type $S$, $k=\ar(\weight)$, and $y_1,\ldots,y_k$ are $k$ pairwise
distinct variables in $\Vars$.
For every ring $S\in\SC$, a \emph{$\Weights$-product of type $S$}
is an expression of the form
\ $t_1 \mal \cdots \mal t_\ell$ \
where $\ell\in\NNpos$ and for each $i\in [\ell]$ either $t_i\in S$ or there exists a
$\weight\in \Weights$ with $\wtype(\weight)=S$ and there exist
$k\deff\ar(\weight)$ pairwise distinct variables $y_1,\ldots,y_k$ in
$\Vars$ such that $t_i=\weight(y_1,\ldots,y_k)$.
By $\varsof(\pweight)$ we denote the set of all variables that occur
in a $\Weights$-product $\pweight$.

\begin{exmp}\label{example:intuition}
Recall Example~\ref{example:WStr}\eqref{exmp:marketplaceStart}--\eqref{exmp:edgeweightsStart},
and let $x$ and $y$ be variables. Examples of $\Weights$-products are
$\mathtt{price}(x) {\cdot} \mathtt{quantity}(x)$,
$\mathtt{embedding}(x)$,
and $\weight(x,y)$.
The logic we will define next is capable of expressing
the following statements.
\begin{enumerate}[(a)]
\item\label{exmp:marketplaceWish}
  Given a first-order formula \(\phi_{\textup{group}}(p)\)
  that defines products of a certain product group based on the structure
  of their transactions, we can
  describe the amount of money a consumer $c$ paid on the specified
  product group via the \emph{$\SC$-term}
  \[t_{\textup{spending}}(c) \deff\ \Ssum{\mathtt{price}(i) \cdot \mathtt{quantity}(i)\;}{\;\exists p\,\exists r\,\big(\phi_{\textup{group}}(p) \land T(i,c,p,r)\big)}.\]
  This term associates with every consumer $c$ the sum of the product
  of $\mathtt{price}(i)$ and $\mathtt{quantity}(i)$ for all
  transaction identifiers $i$ for which there exists a product $p$ and
  a retailer $r$ such that
  the tuple $(i,c,p,r)$ belongs to the transaction table and
  $\phi_{\textup{group}}(p)$ holds.
  The $\SC$-term
  \[t_{\textup{sales}} \deff \ \Ssum{\mathtt{price}(i) \cdot \mathtt{quantity}(i)\;}{\;\exists c\,\exists p\,\exists r\,\big(\phi_{\textup{group}}(p) \land T(i,c,p,r)\big)}\]
  specifies the amount \emph{all} customers have paid on products
  from the product group.

  We might want to select the ``heavy hitters'', i.e.\ all customers
  $c$ for whom $t_{\textup{spending}}(c) > 0.01 \,{\cdot}\, t_{\textup{sales}}$ holds. In
our logic, this is expressed by the formula
\[
\Pred_{>}(t_{\textup{spending}}(c),0.01\,{\cdot}\,t_{\textup{sales}})
\]
where
$\Pred_{>}$ is a predicate name of type $(\QQ,+,\mal)\times (\QQ,+,\mal)$ with
$\sem{\Pred_{>}}=\setc{(r,s)\in\QQ^2}{r>s}$.

\item\label{exmp:social-networkWish}
For vectors $u,v\in\RR^k$, let $d(u,v)$ denote the Euclidean distance between $u$ and $v$.
We might want to use a formula $\phi_{\textup{similar}}(x,y)$
expressing that the two $k$-dimensional vectors associated with
persons $x$ and $y$ have Euclidean distance at most $1$. To express this in our logic,
we can add the rational field $(\QQ,+,\mal)$ to the collection $\SC$ and
use a predicate name $\Pred_{\textup{ED}}$ of arity 3 and type $(\RR^k,+)\times
(\RR^k,+)\times (\QQ,+,\mal)$ with
$\sem{\Pred_{\textup{ED}}}=\setc{(u,v,q)\in\RR^k\times \RR^k\times\QQ}{d(u,v)\leq q}$.
Then, \[\phi_{\textup{similar}}(x,y) \deff
\Pred_{\textup{ED}}(\mathtt{embedding}(x),\mathtt{embedding}(y),1)\]
is a formula
with the desired meaning.

\item\label{exmp:edgeweightsWish}
For each vertex $x$, the sum of the weights of edges
between $x$ and its blue neighbours
is specified by the $\SC$-term
$t_B(x)\deff \Ssum{\weight(x',y)}{(x'{=}x\land E(x',y) \land B(y))}$.
\end{enumerate}
\smallskip

\noindent
We have designed the definition of the syntax of our logic in a way
particularly suitable for formulating and proving the locality results
that are crucial for obtaining our learning results. To obtain a more
user-friendly syntax, i.e.\ which allows to read and construct formulas
in a more intuitive way, one could of course introduce syntactic sugar that
allows to explicitly write statements of the form
\begin{itemize}
\item
$t_{\textup{spending}}(c) > 0.01 \,{\cdot}\, t_{\textup{sales}}$
\ instead of \
$\Pred_{>}(t_{\textup{spending}}(c),0.01\,{\cdot}\,t_{\textup{sales}})$

\item
$d(\textup{embedding}(x),\textup{embedding}(y))\leq 1$
\ instead of \
$\Pred_{\textup{ED}}(\textup{embedding}(x),\textup{embedding}(y),1)$

\item
${\displaystyle\sum_y}\weight(x,y).(E(x,y) \land B(y))$
\ instead of \ $\Ssum{\weight(x',y)}{(x'{=}x\land E(x',y) \land B(y))}$.
\end{itemize}
\end{exmp}

\noindent
We now define the precise
syntax and semantics of our
weight aggregation logic.

\begin{defn}[$\WAPS$]\label{def:WA}
For $\RLPS$, the set of \emph{formulas} and \emph{$\SC$-terms} is built
according to the following rules:
\begin{enumerate}[(1)]
  \item\label{item:atomic} $x_1{=}x_2$ and $R(x_1,\ldots,x_{\ar(R)})$
    are \emph{formulas},\\ where $R\in\sigma$ and
    $x_1,x_2,\ldots,x_{\ar(R)}$ are variables\footnote{In particular,
      if $\ar(R)=0$, then $R()$ is a formula.}.
  \item\label{item:wsimple}
     If $\weight\in\Weights$, $S= \wtype(\weight)$, $s\in S$,
     $k=\ar(\weight)$, and $\ov{x}=(x_1,\ldots,x_k)$ is a tuple of $k$
     pairwise distinct variables, then $( s= \weight(\ov{x}) )$ is a
     \emph{formula}.
  \item\label{item:bool} If $\varphi$ and $\psi$ are formulas,
    then $\nicht\varphi$ and $(\varphi\oder\psi)$ are also \emph{formulas}.
  \item\label{item:exists} If $\varphi$ is a formula and
    $y\in\Vars$, then $\exists y\,\varphi$ is a \emph{formula}.
  \item\label{item:finitegroup}
    If $\phi$ is a formula, $\weight\in\Weights$, $S=
    \wtype(\weight)$, $s\in S$, $k=\ar(\weight)$, and
    $\ov{y}=(y_1,\ldots,y_k)$ is a tuple of $k$ pairwise distinct
    variables, then $\big( s=\sum \weight(\ov{y}).\phi \big)$ is a
    \emph{formula}.

  \item\label{item:Q} If $\Pred\in\PP$, $m= \ar(\Pred)$, and
    $t_1,\ldots,t_m$ are $\SC$-terms such that $(\wtype(t_1),\ldots,\wtype(t_m))=\ptype(\Pred)$, then
    $\Pred(t_1,\ldots,t_m)$ is a \emph{formula}.
  \item\label{item:constterm} For every $S\in\SC$ and every $s\in S$,
    $s$ is an \emph{$\SC$-term} of type $S$.
  \item\label{item:wsimpleterm} For every $S\in\SC$, every $\weight\in\Weights$ of type $S$,
    and every tuple $(x_1,\ldots,x_k)$ of $k\deff\ar(\weight)$ pairwise distinct variables
    in $\Vars$, $\weight(x_1,\ldots,x_k)$ is an \emph{$\SC$-term} of type $S$.
  \item\label{item:plustimesterm} If $t_1$ and $t_2$ are
    $\SC$-terms of the same type $S$, then so are $(t_1 \plus t_2)$
    and $(t_1 \minus t_2)$;
    furthermore, if $S$ is a ring (and not just an abelian group),
    then also $(t_1 \mal t_2)$ is an $\SC$-term of type $S$.
  \item\label{item:countterm} If $\varphi$ is a formula,
    $S\in\SC$, and $\pweight$ is a $\Weights$-product of type $S$, then
    $\Ssum{p}{\varphi}$ is an
    \emph{$\SC$-term} of type $S$.
\end{enumerate}

\noindent
Let $\I=(\A,\beta)$ be a $(\sigma,\Weights)$-interpretation.
For every
formula or $\SC$-term $\xi$ of $\RLPS$, the semantics
$\sem{\xi}^\I$ is defined as follows.

\begin{enumerate}[(1)]
\item
 $\sem{x_1{=}x_2}^{\I}=1$ if $a_1{=}a_2$, and $\sem{x_1{=}x_2}^{\I}=0$
 otherwise; \\
    $\sem{R(x_1,\ldots,x_{\ar(R)})}^{\I}=1$ if
    $(a_1,\ldots,a_{\ar(R)})\in R^\A$, and
    $\sem{R(x_1,\ldots,x_{\ar(R)})}^{\I}=0$ otherwise;
    \\
    where $a_j \deff \beta(x_j)$ for $j\in \set{1,\ldots,\max\set{2,\ar(R)}}$.

\item
    $\sem{( s=\weight(\ov{x}))}^{\I}=1$ if $s= \weight^\A(\beta(x_1),\ldots,\beta(x_k))$,
    and $\sem{( s=\weight(\ov{x}))}^{\I}=0$ otherwise.

\item $\sem{\nicht\varphi}^{\I}=1-\sem{\phi}^{\I}$ \ and \
    $\sem{(\phi\oder\psi)}^{\I}=
    \max\set{\sem{\phi}^{\I},\sem{\psi}^{\I}}$.

\item
    $\sem{\exists
      y\,\varphi}^{\I}=\max\setc{\sem{\phi}^{\I\frac{a}{y}}}{a\in A}$.

\item
    $\sem{\big( s=\sum \weight(\ov{y}).\phi\big)}^{\I}=1$ if $s=\sum_S \setc{\weight^\A(\ov{a})}{\ov{a}=(a_1,\ldots,a_k)\in A^k \text{ with } \sem{\phi}^{\I\frac{a_1,\ldots,a_k}{y_1,\ldots,y_k}}=1}$
(as usual, by convention, we let $\sum_S X = \nullS$ if $X=\emptyset$).

\item
    $\sem{\Pred(t_1,\ldots,t_m)}^{\I}=1$ if
    $\big(\sem{t_1}^{\I},\ldots,\sem{t_m}^{\I}\big)\in\sem{\Pred}$, and
    $\sem{\Pred(t_1,\ldots,t_m)}^{\I}=0$ otherwise.

\item $\sem{s}^{\I}=s$.

\item $\sem{\weight(x_1,\ldots,x_k)}^\I=\weight^\A(\beta(x_1),\ldots,\beta(x_k))$.

\item $\sem{(t_1 \ast t_2)}^{\I}= \sem{t_1}^{\I} \ast_S \sem{t_2}^{\I}$,
  for $\ast\in\set{\plus,\minus,\mal}$.

\item
    $\sem{\Ssum{\pweight}{\varphi}}^{\I}= \sum_S \setc{ \sem{\pweight}^{\I\frac{a_1,\ldots,a_k}{y_1,\ldots,y_k}}
       }{a_1,\ldots,a_k\in A \text{ with }
      \sem{\phi}^{\I\frac{a_1,\ldots,a_k}{y_1,\ldots,y_k}}=1 }$,
   \ where \linebreak[4] $\set{y_1,\ldots,y_k}=\varsof(\pweight)$ and
    $k=|\varsof(\pweight)|$ and $\sem{\pweight}^\I=\sem{t_1}^\I \malS \cdots
    \malS \sem{t_\ell}^\I$ if $\pweight= t_1\mal\cdots\mal t_\ell$
    is of type $S$.
\end{enumerate}
\end{defn}

An \emph{expression} is a formula or an $\SC$-term.
 As usual, for a formula $\phi$ and a $(\sigma,\Weights)$-interpretation
 $\I$, we will often write
 $\I\models\phi$ to indicate that $\sem{\phi}^\I =1$. Accordingly,
 $\I\not\models\phi$ indicates that $\sem{\phi}^{\I}=0$.

The set $\varsof(\xi)$ of an expression $\xi$ is defined as
the set of all variables in $\Vars$ that occur in $\xi$.
The \emph{free variables} $\free(\xi)$ of $\xi$ are defined as follows:
$\free(\xi)=\varsof(\xi)$ if $\xi$ is built according to one of the
rules~\eqref{item:atomic}, \eqref{item:wsimple}, \eqref{item:constterm},
\eqref{item:wsimpleterm}; $\free(\nicht\phi)=\free(\phi)$,
$\free((\phi\oder\psi))=\free(\phi)\cup\free(\psi)$,
$\free(\exists y\,\phi)=\free(\phi)\setminus\set{y}$,
$\free((s=\sum\weight(y_1,\ldots,y_k).\phi))=
  \free(\phi)\setminus\set{y_1,\ldots,y_k}$;
$\free(\Pred(t_1,\ldots,t_m))=\bigcup_{i=1}^m\free(t_i)$;
$\free((t_1\ast t_2))=\free(t_1)\cup\free(t_2)$ for
  $\ast\in\set{\plus,\minus,\mal}$;
$\free(\sum \pweight.\phi)=\free(\phi)\setminus\varsof(\pweight)$.
As usual, we will write $\xi(\ov{x})$ for $\ov{x}=(x_1,\ldots,x_k)$ to
indicate that $\free(\xi)\subseteq\set{x_1,\ldots,x_k}$.
A \emph{sentence} is a $\WAPS$-formula $\phi$ with
$\free(\phi)=\emptyset$.
A \emph{ground $\SC$-term} is an $\SC$-term $t$ of $\WAPS$ with
$\free(t)=\emptyset$.

For a
$(\sigma,\Weights)$-structure $\A$ and a tuple
$\ov{a}=(a_1,\ldots,a_k)\in A^k$, we write $\A\models\phi[\ov{a}]$ or
$(\A,\ov{a})\models \phi$ to indicate that for every assignment
$\beta \colon \Vars\to A$ with $\beta(x_i)=a_i$ for all $i\in[k]$, we
have
$\I\models\phi$,
for $\I=(\A,\beta)$.
Similarly, for an $\SC$-term $t(\ov{x})$ we write $t^\A[\ov{a}]$ to denote $\sem{t}^\I$.

\begin{defn}[$\WAun$ and $\FOW$]\label{def:WAun}\label{def:FOW}
The set of \emph{formulas} and \emph{$\SC$-terms} of the logic
$\WAunPS$ is built according to the same rules as for the logic $\WAPS$,
with the following restrictions:
\begin{description}
\item[\eqref{item:finitegroup}$_1$:] rule~\eqref{item:finitegroup} can only be applied if $S$ is
  \emph{finite},
\item[\eqref{item:Q}$_1$:] rule~\eqref{item:Q} can only be applied if
 $|\free(t_1)\cup \cdots \cup \free(t_m)|\leq 1$.
\end{description}
$\FOWS$ is the restriction of $\WAunPS$ where
rule~\eqref{item:countterm} cannot be applied.
\end{defn}

Note that first-order logic $\FO[\sigma]$ is the restriction of $\FOWS$
where only rules~\eqref{item:atomic}, \eqref{item:bool}, and
\eqref{item:exists} can be applied.
As usual, we write $(\phi\und\psi)$ and $\forall y\,\phi$ as shorthands
for $\nicht(\nicht\phi \oder\nicht\psi)$ and $\nicht\exists y\,\nicht\phi$.
The \emph{quantifier rank} $\qr(\xi)$ of a $\WAPS$-expression $\xi$ is
defined as the maximum nesting depth of constructs using
rules~\eqref{item:exists} and~\eqref{item:finitegroup}
in order to construct $\xi$.
The \emph{aggregation depth} $\countr(\xi)$ of $\xi$ is defined as
the maximum nesting depth of term constructions using
rule~\eqref{item:countterm} in order to construct $\xi$.

\begin{rem}\label{remark:FOunCinWAun}
$\FOW$ can be viewed as an extension of first-order logic with
modulo-counting quantifiers:
if \(\SC\) contains the abelian group \((\ZZ/m\ZZ,+)\) for some
$m\geq 2$,
and \(\Weights\) contains a unary weight symbol \(\weightone_m\)
  of type $\ZZ/m\ZZ$
  such that \(\weightone_m^\A(a) = 1\) for all \(a \in A\),
  then the modulo $m$ counting quantifier
  \(\exists^{i \text{ mod }m}y\, \phi\) (stating that the number of
  interpretations for $y$ that satisfy $\phi$ is congruent to $i$ modulo $m$)
  can be expressed in $\FOWS$ via
  \(\big(i=\Ssum{\weightone_m(y)}{\phi}\big)\).
\smallskip

\noindent
$\WAun$ can be viewed as an extension of the logic $\FOunC$
of~\cite{GroheSchweikardt_FOunC}:
if $\SC$ contains the integer ring $(\ZZ, +, \mal)$ and $\Weights$
contains a unary weight symbol $\weightone$ of type $\ZZ$ such that
\(\weightone^\A(a) = 1\) for all \(a \in A\) on all considered
$(\sigma,\Weights)$-structures $\A$,
then the counting term \(\#(y_1,\ldots,y_k).\phi\) of $\FOunC$ (which
counts the number of tuples $(y_1,\ldots,y_k)$ that satisfy $\phi$) can
be expressed in $\WAunPS$ via the $\SC$-term
  \(\Ssum{\pweight}{\phi}\) for $\pweight\deff \weightone(y_1) \mal
  \cdots \mal \weightone(y_k)$.
\smallskip

\noindent
Let us mention, again, that we have designed the precise definition of the syntax of our logic in a way
particularly suitable for formulating and proving the locality results
that are crucial for obtaining our learning results. To obtain a more
user-friendly syntax, i.e.\ which allows to read and construct formulas
in a more intuitive way, it would of course make sense to introduce syntactic sugar that
allows to explicitly write statements of the form
\begin{itemize}
\item
\(\#(y_1,\ldots,y_k).\phi\)
\ instead of \
  \(\Ssum{\pweight}{\phi}\) for $\pweight\deff \weightone(y_1) \mal
  \cdots \mal \weightone(y_k)$
\item
$\big(\#(y).\phi \equiv i \ \text{mod} \ m\big)$ \ or \
\(\exists^{i \text{ mod }m}y\, \phi\)
\ instead of \
\(\big(i=\Ssum{\weightone_m(y)}{\phi}\big)\).
\end{itemize}
For this, one would tacitly assume that $\SC$ contains
$(\ZZ,+,\mal)$ (or \((\ZZ/m\ZZ,+)\)) and
$\Weights$ contains a unary weight symbol $\weightone$ of type $\ZZ$
(or $\weightone_m$ of type $\ZZ/m\ZZ$)
where $\weightone^{\A}(a)=1$ ($= \weightone^{\A}_m(a)$) for every $a\in A$ and every considered
$(\sigma,\Weights)$-structure $\A$.
\end{rem}

To close this section, we return to the running examples from
Examples~\ref{example:WStr} and
\ref{example:intuition}.

\begin{exmp}\label{example:moreintuition}
We use the syntactic sugar introduced at the end of Remark~\ref{remark:FOunCinWAun}.
\begin{enumerate}[(a)]
\item\label{exmp:marketplace}
  The number of consumers who bought products $p$ from the product
  group defined by \(\phi_{\textup{group}}(p)\) is specified by the
  $\SC$-term
  \[t_{\textup{\#cons}} \deff \
   \Ssum{\mathtt{one}(c)\;}{\;\exists i\,\exists p\,\exists
   r\,(\phi_{\textup{group}}(p) \land T(i,c,p,r))};\] and using the
 syntactic sugar described above, this $\SC$-term can be expressed via
 $
  \#(c).{\;\exists i\,\exists p\,\exists
   r\,(\phi_{\textup{group}}(p) \land T(i,c,p,r))}
 $.

  The consumers $c$ who spent at least as much as
  the \emph{average consumer} on the products $p$
  satisfying \(\phi_{\textup{group}}(p)\) can be
  described by the formula
  \[\phi_{\textup{spending}}(c) \deff \ \P_{\geq}\big((t_{\textup{spending}}(c)
  \,\mal\, t_{\textup{\#cons}})\, ,\,  t_{\textup{sales}} \big),\]
  where $\P_{\geq}$ is a binary predicate in $\Ps$ of type
  $\QQ\times\QQ$ that is interpreted by the $\geq$-relation.
 To improve readability, one could introduce syntactic sugar that
 allows to express this as
$t_{\textup{spending}}(c) \geq t_{\textup{sales}} /
t_{\textup{\#cons}}$.
The formula $\phi_{\textup{spending}}(c)$ belongs to \(\WAunPS\).
\smallskip

\item\label{exmp:social-network}
The term $t_{\#\textup{follows}}(x)\deff \#(y).F(x,y)$ specifies
the number of users $y$ followed by person $x$.
The term $t_{\textup{sum}}(x) \deff
\Ssum{\mathtt{embedding(y)}}{F(x,y)}$ specifies the sum of the vectors
associated with all users $y$ followed by $x$.
To describe the users $x$ whose embedding
is $\delta$-close (for some fixed $\delta>0$) to the average of the embeddings of users they follow\footnote{Depending on the target of the embeddings,
  this could mean that the user mostly follows users
  with a very similar personality or political leaning.}, we
might want to use a formula $\phi_{\textup{close}}(x)$ of the form
\[
   d\,\big(\,\mathtt{embedding}(x) \;,\;
   {\textstyle\frac{1}{t_{\#\textup{follows}}(x)}}\,{\cdot}\,
   t_{\textup{sum}}(x) \,\big) \ < \ \delta\,.
\]
We can describe this in
$\WAunPS$ by the formula
\[
 \phi_{\textup{close}}(x) \ \deff \
 \Pred_{\dist < \delta}(\mathtt{embedding}(x),t_{\#\textup{follows}}(x),t_{\textup{sum}}(x)),
\]
where
\(\Pred_{\dist < \delta}\) is a ternary predicate in $\PP$ of type
  \(\RR^k \times \ZZ \times \RR^k\)
  consisting of all triples \((\bar{v}, \ell, \bar{w})\)
  with \(\ell > 0\) and
$d(\bar{v},\frac{1}{\ell}{\cdot}\bar{w})<\delta$.

\item\label{exmp:edgeweights}
Recall the term $t_B(x)$ introduced in
Example~\ref{example:intuition}\,\eqref{exmp:edgeweightsWish} that
specifies the sum of the weights of edges between $x$ and its blue
neighbours, and let $t_R(x)$ be a similar term summing up the
weights of edges between $x$ and its red neighbours (using
the syntactic sugar introduced at the end of
Example~\ref{example:intuition}, this can be described as
${\displaystyle\sum_y}\weight(x,y).(E(x,y) \land R(y))$).
To specify the vertices $x$ that have exactly 5 red neighbours, we can
use the formula $\phi_{\textup{5\,red}}(x) \deff
(\,5=\#(y).(E(x,y)\und R(y)) \,)$.
Let us now assume we are given a particular set $H\subseteq
\RR^{2k}$ and we
want to specify the vertices $x$ that have exactly
5 red neighbours and for which, in addition, the
$2k$-ary vector obtained by concatenating the $k$-ary vectors
computed by summing up the weights of edges
between $x$ and its blue neighbours
and by summing up the weights of edges
between $x$ and its red neighbours
belongs to $H$. To express this, we can use
a binary predicate $\Pred$ of type
$\mathbb{R}^k\times\mathbb{R}^k$ with
$\sem{\Pred} = \big\{(\bar{u},\bar{v}) \in \mathbb{R}^k \times
\mathbb{R}^k \,:\, (u_1,\ldots,u_k,v_1,\ldots,v_k)\in H\}$.
Then, the \(\WAunPS\)-formula
$\psi(x) \deff \phi_{\textup{5\,red}}(x) \land \Pred(t_B(x),t_G(x))$
specifies the vertices $x$ we are interested in.
\end{enumerate}
\end{exmp}

\section{Locality Properties of $\FOW$ and $\WAun$}\label{subsec:Locality}\label{sec:Locality}
We now summarise locality properties of $\FOW$ and $\WAun$ that are
similar to well-known locality properties of first-order logic $\FO$ and to
locality properties of $\FOunC$ achieved
in~\cite{GroheSchweikardt_FOunC}. This includes
\emph{Feferman-Vaught decompositions} (Section~\ref{sec:FV}) and
a \emph{Gaifman normal form} for $\FOW$ (Section~\ref{sec:Gaifman}),
and a \emph{localisation theorem} for the more expressive logic $\WAun$ (Section~\ref{sec:Decomposition}).

For the remainder of this section, let us fix a signature $\sigma$, a
collection $\SC$ of rings and/or abelian groups, a finite set
$\Weights$ of weight symbols, and an $\SC$-predicate collection $(\PP,\ar,\ptype,\sem{{\cdot}})$.

The notion of \emph{local formulas} is defined as usual~\cite{Lib04}:
let $r\in\NN$.
A $\WAPS$-form\-ula $\phi(\ov{x})$ with free variables
$\ov{x}=(x_1,\ldots,x_k)$ is \emph{$r$-local (around
  $\ov{x}$)}
if for every $(\sigma,\Weights)$-structure $\A$ and all
$\ov{a}\in A^k$, we have
\ $\A\models\phi[\ov{a}]
  \iff
 \NrA{\ov{a}}\models\phi[\ov{a}]$\,.
A formula is \emph{local} if it is $r$-local for some $r\in\NN$.

For an $r\in\NN$, it is
straightforward to construct an $\FO[\sigma]$-formula
$\dist^{\sigma}_{\leq r}(x,y)$ such that for every $(\sigma,\Weights)$-structure $\A$
and all $a,b\in A$, we have
$\A\models\dist^{\sigma}_{\leq r}[a,b]$ $\iff$
$\dist^{\A}(a,b)\leq r$.
To improve readability, we write
$\dist^\sigma(x,y)\,{\leq}\, r$ for
$\dist^{\sigma}_{\leq r}(x,y)$, and
$\dist^\sigma(x,y)\,{>}\,r$ for $\nicht\dist^{\sigma}_{\leq r}(x,y)$;
and we omit the superscript $\sigma$ when it is clear from the context.
For a tuple $\ov{x}=(x_1,\ldots,x_k)$ of variables,
$\dist(\ov{x},y)\,{>}\,r$ is a shorthand for
$\Und_{i=1}^k\dist(x_i,y)\,{>}\,r$, and $\dist(\ov{x},y)\,{\leq}\,r$
is a shorthand for $\Oder_{i=1}^k\dist(x_i,y)\,{\leq}\,r$.
For $\ov{y}=(y_1,\ldots,y_\ell)$, we use $\dist(\ov{x};\ov{y})\,{>}\,r$
and $\dist(\ov{x};\ov{y})\,{\leq}\,r$
as shorthands for $\Und_{j=1}^\ell \dist(\ov{x},y_j)\,{>}\,r$ and
$\Oder_{j=1}^\ell \dist(\ov{x},y_j)\,{\leq}\,r$, respectively.

The \emph{$r$-localisation} $\phi^{(r)}$ of a $\WAPS$-formula $\phi(\ov{x})$
is the formula obtained from $\phi$ by replacing
every subformula of the form $\exists y\,\phi'$ with the formula
$\exists y\, \big(\phi' \und \dist(\ov{x},y)\leq r \big)$,
replacing
every subformula of the form $\big(s=\sum \weight(\ov{y}).\phi'\big)$,
for $\ov{y}=(y_1,\ldots,y_k)$,
with the formula
$\big(s=\sum \weight(\ov{y}).(\phi'\und\Und_{j=1}^k\dist(\ov{x},y_j)\leq r)\big)$,
and replacing
every $\SC$-term of the form $\Count{\pweight}{\phi'}$
with the $\SC$-term
$\Count{\pweight}{\big(\phi'\und\Und_{j=1}^k\dist(\ov{x},y_j)\leq
  r\big)}$,
where $\set{y_1,\ldots,y_k}=\free(\phi')$.
The resulting formula $\phi^{(r)}(\ov{x})$ is $r$-local.

\subsection{Feferman-Vaught Decomposition for  $\FOW$}
\label{sec:FV}\label{appendix:FV}
We pick two new unary relation symbols $X,Y$
that do not belong to $\sigma$, and we let
$\sigma'\deff\sigma\cup\set{X,Y}$.

\begin{defn}\label{def:disjsum}
Let $\A,\B$ be $(\sigma,\Weights)$-structures with $A\cap
B=\emptyset$.
The \emph{disjoint sum} $\A\dsum\B$ is the
$(\sigma',\Weights)$-structure $\C$ with universe $C=A\cup B$,
$X^\C=A$, $Y^\C=B$, $R^\C=R^\A\cup R^\B$ for all $R\in\sigma$, and such that for
all $\weight\in\Weights$ and $k\deff\ar(\weight)$ and all
$\ov{c}=(c_1,\ldots,c_k)\in C^k$, we have
$\weight^\C(\ov{c})=\weight^\A(\ov{c})$ if $\ov{c}\in A^k$,
$\weight^\C(\ov{c})=\weight^\B(\ov{c})$ if $\ov{c}\in B^k$, and
$\weight^\C(\ov{c})=\nullS$ otherwise (for $S\deff\wtype(\weight)$).
The \emph{disjoint union} $\A\dunion\B$ is the
$(\sigma,\Weights)$-structure obtained from $\C\deff \A\dsum\B$ by omitting the
relations $X^\C, Y^\C$.
\end{defn}

\begin{defn}\label{def:FVZ}
Let $\LogicL$ be a subset of $\WAPS$.

\noindent
Let $k,\ell\in\NN$ and let $\ov{x}=(x_1,\ldots,x_k)$,
$\ov{y}=(y_1,\ldots,y_\ell)$ be tuples of $k{+}\ell$ pairwise distinct variables.
Let $\phi$ be a $\WAPSR{\PP}{\sigma',\SC,\Weights}$-formula with
$\free(\phi)\subseteq \set{x_1,\ldots,x_k,y_1,\ldots,y_\ell}$.
A \emph{Feferman-Vaught decomposition}
of $\phi$ in $\LogicL$ w.r.t.\ $(\ov{x};\ov{y})$ is a finite, non-empty set
$\Delta$ of tuples of the form
$\big(\alpha,\beta\big)$ where $\alpha,\beta\in\LogicL$ and
$\free(\alpha)\subseteq\set{x_1,\ldots,x_k}$ and
$\free(\beta)\subseteq\set{y_1,\ldots,y_\ell}$, such that
the following is true
for all $(\sigma,\Weights)$-structures $\A,\B$ with $A\cap
B=\emptyset$ and all $\ov{a}\in A^k$, $\ov{b}\in B^\ell$:
$\A\dsum\B\models \phi[\ov{a},\ov{b}]$ $\iff$ there exists
$(\alpha,\beta)\in\Delta$ such that $\A\models\alpha[\ov{a}]$ and $\B\models\beta[\ov{b}]$.
\end{defn}

Our first main result provides Feferman-Vaught decompositions for
$\FOW$.

\begin{thm}[Feferman-Vaught decompositions for
  $\FOWS$]\label{thm:FVforFOW} \ \\
Let $k,\ell\in\NN$ and let $\ov{x}=(x_1,\ldots,x_k)$,
$\ov{y}=(y_1,\ldots,y_\ell)$ be tuples of $k{+}\ell$ pairwise distinct variables.
For every $\FOWSR{\PP}{\sigma',\SC,\Weights}$-formula $\phi$ with
$\free(\phi)\subseteq \set{x_1,\ldots,x_k,y_1,\ldots,y_\ell}$, there
exists a Feferman-Vaught decomposition $\Delta$ in $\LogicL$ of $\phi$
w.r.t.\ $(\ov{x};\ov{y})$, where
$\LogicL\deff\LogicL_\phi$ is the class of all $\FOWS$-formulas of quantifier rank
at most $\qr(\phi)$ which
use only those $\P\in\Ps$ and $S\in\SC$ that occur in $\phi$ and
only those $\SC$-terms that occur in $\phi$ or that are of
the form $s$ for an $s\in S\in\SC$ where $S$ is finite and occurs in $\phi$.

Furthermore, there is an algorithm that
computes $\Delta$ upon input of $\phi,\ov{x},\ov{y}$.
\end{thm}

The proof proceeds in a
similar way as the proof of the Feferman-Vaught decomposition for
first-order logic with modulo-counting quantifiers
in~\cite{KuskeSchweikardt_Gaifman}.
Before presenting the theorem's proof, let us formulate a
straightforward corollary of
Theorem~\ref{thm:FVforFOW}.

\begin{cor}\label{cor:FV}
Let $k,\ell\in\NN$ and let $\ov{x}=(x_1,\ldots,x_k)$,
$\ov{y}=(y_1,\ldots,y_\ell)$ be tuples of $k{+}\ell$ pairwise distinct variables.
Upon input of an $r\in\NN$ and an $r$-local $\FOWS$-formula
$\phi(\ov{x},\ov{y})$, one can compute a finite, non-empty set
$\Delta$ of pairs $\big(\alpha(\ov{x}),\beta(\ov{y})\big)$ of
$\LogicL$-formulas,
where
$\LogicL$ is the class of all $r$-localisations of formulas in the
class $\LogicL_\phi$ of Theorem~\ref{thm:FVforFOW},
such that the following two formulas are equivalent:
\begin{itemize}
\item
  $\Big(\Und_{i=1}^k\Und_{j=1}^\ell \dist(x_i,y_j)>2r{+}1 \Big)
   \ \und \ \phi(\ov{x},\ov{y})$
\item
  $\Big(\Und_{i=1}^k\Und_{j=1}^\ell \dist(x_i,y_j)>2r{+}1 \Big)
   \ \und \ \Oder_{(\alpha,\beta)\in\Delta} \big( \alpha(\ov{x})\und\beta(\ov{y})\big)$.
\end{itemize}
\end{cor}

\bigskip

\noindent
The remainder of Section~\ref{sec:FV} is devoted to the proofs of
Theorem~\ref{thm:FVforFOW} and Corollary~\ref{cor:FV}.

\begin{proof}[Proof of Theorem~\ref{thm:FVforFOW}] \ \\
We proceed by induction on the construction of $\phi$, and we use
an arbitrary unsatisfiable formula $\bot$ (e.g.\
$\bot\deff \exists z\,\nicht z{=}z$) and an arbitrary tautology $\top$
(e.g., $\top\deff\nicht\bot$).

For the induction base, we consider formulas built according to the rules
\eqref{item:atomic}, \eqref{item:wsimple}, and
\eqref{item:Q}$_1$ of Definitions~\ref{def:WA} and \ref{def:WAun}.
\\
Rule \eqref{item:atomic} can be handled in exactly the same way as in
the traditional Feferman-Vaught construction for first-order logic $\FO$
(cf., e.g., \cite{FefV59,DBLP:journals/apal/Makowsky04,Gro08}).
\\
For rule \eqref{item:wsimple}, let $\phi$ be of the form
$\big(s=\weight(z_1,\ldots,z_m)\big)$.
If $\set{z_1,\ldots,z_m}\subseteq \set{x_1,\ldots,x_k}$,
we can choose $\Delta\deff\set{(\phi,\top)}$.
If $\set{z_1,\ldots,z_m}\subseteq\set{y_1,\ldots,y_\ell}$,
we can choose $\Delta\deff\set{(\top,\phi)}$.
Otherwise, we know that $\set{z_1,\ldots,z_m}$ contains variables from
$\ov{x}$ and variables from $\ov{y}$; and if $s=\nullS$, we can choose
$\Delta\deff\set{(\top,\top))}$, and otherwise, we can choose
$\Delta\deff\set{(\bot,\bot)}$. It is straightforward to verify that
$\Delta$ is a Feferman-Vaught decomposition in $\FOWS$ of $\phi$
w.r.t.\ $(\ov{x};\ov{y})$.
\\
For rule \eqref{item:Q}$_1$, let $\phi$ be of the form
$\P(t_1,\ldots,t_m)$, where $\P\in\Ps$ and $t_1,\ldots,t_m$
are $\SC$-terms.
We know that each $t_i$ is built using the rules
\eqref{item:constterm}--\eqref{item:plustimesterm}, and that
there is one variable $z$ such that
$\varsof(t_i)\subseteq\set{z}$ for all $i\in[m]$.
Thus, if $z\in\set{x_1,\ldots,x_k}$, we can choose $\Delta\deff\set{(\phi,\top)}$;
and if $z\in\set{y_1,\ldots,y_\ell}$,
we can choose $\Delta\deff\set{(\top,\phi)}$.

For the induction step, we consider formulas built according to the
rules \eqref{item:bool}, \eqref{item:exists}, and \eqref{item:finitegroup}$_1$
of Definitions~\ref{def:WA} and \ref{def:FOW}.
Rules \eqref{item:bool} and \eqref{item:exists} can be handled in exactly
the same way as for first-order logic (cf., e.g.,
\cite{FefV59,DBLP:journals/apal/Makowsky04,Gro08}).
For rule \eqref{item:finitegroup}$_1$, we proceed in a similar way as the
case of modulo-counting quantifiers was handled in
\cite{KuskeSchweikardt_Gaifman}: Let $\phi$ be of the form
$\big(s= \sum \weight(\ov{z}).\psi \big)$, for a tuple of variables
$\ov{z}=(z_1,\ldots,z_m)$ and a weight symbol $\weight\in\Weights$
whose type $S\deff\wtype(\weight)$ is \emph{finite}.
For every $i\in S$, let
\[
 \chi_i \ \deff \ \Big( i = \sum \weight(\ov{z}). \big(\psi \und
 {\displaystyle \Und_{j=1}^m} X(z_j)\big)\Big)
\qquad\text{and}\qquad
 \theta_i \ \deff \ \Big( i = \sum \weight(\ov{z}). \big(\psi \und
 {\displaystyle \Und_{j=1}^m} Y(z_j)\big)\Big)\,.
\]
Let $I\deff\bigsetc{(i_1,i_2)\in S\times S}{i_1 \plusS i_2 = s}$.
It is straightforward to see that for all $(\sigma,\Weights)$-structures $\A$ and $\B$ with
$A\cap B=\emptyset$ and all $\ov{a}\in A^k$, $\ov{b}\in B^\ell$, we
have:
\begin{equation}
 (\A\dsum\B,\ov{a},\ov{b}) \models \big( s = \sum \weight(\ov{z}).\psi
 \big) \quad \iff \quad
 (\A\dsum\B,\ov{a},\ov{b}) \models
 {\displaystyle\Oder_{(i_1,i_2)\in I}} \big(
   \chi_{i_1} \, \und \, \theta_{i_2}
 \big)\,.
\end{equation}
Since $(\chi_{i_1}\und\theta_{i_2})$ is equivalent to
$\nicht(\nicht\chi_{i_1}\oder\nicht\theta_{i_2})$ and we already know
how to handle formulas built using rule \eqref{item:bool}, we are done
once we have shown the following:
\begin{claim}\label{claim:FV}
For every $i\in S$, one can compute Feferman-Vaught decompositions
$\Delta_{\chi_i}$ and $\Delta_{\theta_i}$ in $\FOWS$ of $\chi_i$
and $\theta_i$ w.r.t.\ $(\ov{x};\ov{y})$.
\end{claim}
To prove the claim, fix an $i\in S$. We show how to construct
$\Delta_{\theta_i}$ (the
construction of $\Delta_{\chi_i}$ is analogous).
By the induction hypothesis, we can construct a Feferman-Vaught
decomposition $\Delta$ in $\FOWS$ of $\psi$ w.r.t.\ $(\ov{x};\ov{y}\ov{z})$.
It is an easy exercise to see that, w.l.o.g., we can assume that \emph{the $\alpha$s in $\Delta$ are
  mutually exclusive}, i.e.\ for every two distinct $(\alpha,\beta)$ and
$(\alpha',\beta')$ in $\Delta$, the formula $(\alpha\und\alpha')$ is
unsatisfiable.
Let
$
 \Delta' \deff
  \left\{\,
   \big(\,
     \alpha\,,\, \big( i={\textstyle\sum} \weight(\ov{z}).\beta \big)
   \; \big)
   \ : \
   (\alpha,\beta)\in\Delta
 \, \right\}
$.
If $i\neq \nullS$, we let $\Delta_{\theta_i}\deff\Delta'$.
If $i=\nullS$, we let
$\Delta_{\theta_i}
 \deff
 \Delta'\cup
  \set{\, \big(\, \Und_{\alpha\in \mathbb{A}}\nicht\alpha\, , \top\,))\,}$,
where $\mathbb{A}\deff\setc{\alpha}{\text{there exists }\beta\text{ such that
  }(\alpha,\beta)\in \Delta}$.

It remains to verify that $\Delta_{\theta_i}$ is a Feferman-Vaught
decomposition of $\theta_i$. Consider arbitrary
$(\sigma,\Weights)$-structures $\A$ and $\B$ with $A\cap B=\emptyset$,
and let $\ov{a}\in A^k$, $\ov{b}\in B^\ell$.
By definition, we have
$\A\dsum\B\models\theta_i[\ov{a},\ov{b}]$ $\iff$
$i=\sum_S\setc{\weight^\B(\ov{c})}{\ov{c}\in M}$, for
 $M\deff \setc{\ov{c}\in
  B^m}{\A\dsum\B\models\psi[\ov{a},\ov{b},\ov{c}]}$.
Since $\Delta$ is a Feferman-Vaught decomposition of $\psi$ w.r.t.\
$(\ov{x};\ov{y}\ov{z})$, we have
$\A\dsum\B\models\psi[\ov{a},\ov{b},\ov{c}]$ $\iff$
there exists $(\alpha',\beta')\in \Delta$ such that
$\A\models\alpha'[\ov{a}]$ and
$\B\models\beta'[\ov{b},\ov{c}]$. Furthermore, we know that the
$\alpha$s in $\Delta$ are mutually exclusive.
Thus, there either is exactly one $\alpha\in\mathbb{A}$ such that
$\A\models\alpha[\ov{a}]$ (we call this \emph{Case~1}), or for all $\alpha\in\mathbb{A}$, we have
$\A\not\models\alpha[\ov{a}]$ (we call this \emph{Case~2}).

In \emph{Case~1}, there is exactly one $\beta$ such that
$(\alpha,\beta)\in \Delta$ (this is implied by our definition of the
notion ``the $\alpha$s are mutually exclusive''). Hence,
$M=\setc{\ov{c}\in B^m}{\B\models\beta[\ov{b},\ov{c}]}$. Thus,
$\A\dsum\B\models\theta_i[\ov{a},\ov{b}]$ $\iff$
$i=\sum_S\setc{\weight^\B(\ov{c})}{\ov{c}\in M}$ $\iff$
$\B\models \big(i=\sum \weight(\ov{z}).\beta \big)[\ov{b}]$ $\iff$
there are $(\hat{\alpha},\hat{\beta})\in\Delta_{\theta_i}$ such that
$\A\models\hat{\alpha}[\ov{a}]$ and $\B\models\hat{\beta}[\ov{b}]$.

In \emph{Case~2}, $M=\setc{\ov{c}\in
  B^m}{\A\dsum\B\models\psi[\ov{a},\ov{b},\ov{c}]} = \setc{\ov{c}\in
  B^m}{\text{there exists } (\alpha',\beta')\in\Delta \text{ such that
  }\A\models\alpha'[\ov{a}] \text{ and }\B\models\beta'[\ov{b}] } =
\emptyset$.
Hence, $\A\dsum\B\models\theta_i[\ov{a},\ov{b}]$ $\iff$ $i=0$ $\iff$
$\Delta_{\theta_i}$ contains the tuple $\big(\Und_{\alpha\in
  \mathbb{A}}\nicht\alpha\,,\,\top\big)$ $\iff$ there are $(\hat{\alpha},\hat{\beta})\in\Delta_{\theta_i}$ such that
$\A\models\hat{\alpha}[\ov{a}]$ and $\B\models\hat{\beta}[\ov{b}]$.

In summary, we obtain that $\Delta_{\theta_i}$ is a Feferman-Vaught
decomposition of $\theta_i$. This completes the proof of
Claim~\ref{claim:FV} and of Theorem~\ref{thm:FVforFOW}.
\end{proof}

\begin{proof}[Proof of \cref{cor:FV}] \ \\
  Let \(\phi\) be an \(r\)-local \(\FOWS\)-formula.
  Using \cref{thm:FVforFOW}, we can compute a
  Feferman-Vaught decomposition \(\Delta'\) in \(\LogicL_\phi\)
  of \(\phi\) w.r.t.\ \((\ov{x};\ov{y})\).
  Let
  \(\Delta \coloneqq \setc{(\alpha^{(r)}, \beta^{(r)}}{(\alpha, \beta) \in \Delta'}\).
  We show that the two formulas
  \[\psi_1(\ov{x}, \ov{y}) \coloneqq \Big(\Und_{i=1}^k\Und_{j=1}^\ell \dist(x_i,y_j)>2r{+}1 \Big)
    \ \und \phi(\ov{x}, \ov{y})\]
  and
  \[\psi_2(\ov{x}, \ov{y}) \coloneqq \Big(\Und_{i=1}^k\Und_{j=1}^\ell \dist(x_i,y_j)>2r{+}1 \Big)
    \ \und \ \Oder_{(\alpha^{(r)},\beta^{(r)})\in\Delta} \big( \alpha^{(r)}(\ov{x})\und\beta^{(r)}(\ov{y})\big)\]
  given in \cref{cor:FV} are equivalent.

  Let \(\A\) be a \((\sigma,\Weights)\)-structure,
  \(\ov{a} \in A^k\), and \(\ov{b} \in A^\ell\).
  If \(\dist(\ov{a},\ov{b}) \leq 2r{+}1\),
  then
  \(\A \not\models \psi_1[\ov{a},\ov{b}]\) and
  \(\A \not\models \psi_2[\ov{a},\ov{b}]\).
  Now let \(\dist(\ov{a},\ov{b}) > 2r{+}1\).
  Then, since \(\phi\) is \(r\)-local,
  \(\A \models \psi_1[\ov{a},\ov{b}]\)
  if and only if
  \(\Neighbr{\A}{\ov{a}} \dunion \Neighbr{\A}{\ov{b}} \models \phi[\ov{a},\ov{b}]\).
  Thus, we obtain
  \begin{align*}
         &\A \models \psi_1[\ov{a},\ov{b}] \\
    \iff &\Neighbr{\A}{\ov{a}} \dunion \Neighbr{\A}{\ov{b}} \models \phi[\ov{a},\ov{b}] \\
    \iff &\Neighbr{\A}{\ov{a}} \oplus \Neighbr{\A}{\ov{b}} \models \phi[\ov{a},\ov{b}] \\
    \iff &\exists (\alpha,\beta) \in \Delta' \colon \Neighbr{\A}{\ov{a}} \models \alpha[\ov{a}]
      \land \Neighbr{\A}{\ov{b}} \models \beta[\ov{b}] \\
    \iff &\exists (\alpha,\beta) \in \Delta' \colon \Neighbr{\A}{\ov{a}} \models \alpha^{(r)}[\ov{a}]
      \land \Neighbr{\A}{\ov{b}} \models \beta^{(r)}[\ov{b}] \\
    \iff &\exists (\alpha,\beta) \in \Delta' \colon \Neighbr{\A}{\ov{a}} \oplus \Neighbr{\A}{\ov{b}}
      \models \alpha^{(r)}[\ov{a}] \land \beta^{(r)}[\ov{b}] \\
    \iff &\exists (\alpha,\beta) \in \Delta' \colon \Neighbr{\A}{\ov{a}} \dunion \Neighbr{\A}{\ov{b}}
      \models \alpha^{(r)}[\ov{a}] \land \beta^{(r)}[\ov{b}] \\
    \iff &\exists (\alpha,\beta) \in \Delta' \colon \A \models \alpha^{(r)}[\ov{a}] \land \beta^{(r)}[\ov{b}] \\
    \iff &\A \models \psi_2[\ov{a},\ov{b}].
  \end{align*}
  We can switch between the disjoint sum and the disjoint union of structures
  because the considered formulas only use relations from the disjoint union.
  All in all, this shows that \(\psi_1 \equiv \psi_2\).
\end{proof}

\subsection{Gaifman Normal Form for $\FOW$}
\label{sec:Gaifman}
\label{appendix:Gaifman}
We now turn to a notion of Gaifman normal form for $\FOW$.

\begin{defn}
\label{def:GNFforFOW}
A \emph{basic-local sentence} in \allowbreak $\FOWS$ is a sentence of
the form \allowbreak
$
  \exists x_1\cdots \exists x_\ell\,\big(
   \Und_{1\leq i<j\leq \ell}\dist(x_i,x_j)>2r
   \, \und \,
   \Und_{i=1}^\ell \lambda(x_i)
  \big)
$, \
where $\ell\in\NNpos$, $r\in\NN$, $\lambda(x)$ is an $r$-local
$\FOWS$-formula, and $x_1,\ldots,x_\ell$ are $\ell$ pairwise distinct
variables.

A \emph{local aggregation sentence} in $\FOWS$ is a sentence
of the form
$
\big(\, s=
 \sum \weight(\ov{y}).\lambda(\ov{y})
\,\big)
$, \
where $\weight\in\Weights$, $s\in S\deff\wtype(\weight)$, $\ell=\ar(\weight)$,
$\ov{y}=(y_1,\ldots,y_\ell)$ is a tuple of $\ell$ pairwise distinct
variables, and $\lambda(\ov{y})$ is an $r$-local $\FOWS$-formula.

A $\FOWS$-formula \emph{in Gaifman normal form} is a Boolean
combination of local $\FOWS$-formulas, basic-local sentences in
$\FOWS$, and local aggregation sentences in $\FOWS$.
\end{defn}

Our next main theorem provides a Gaifman normal form for $\FOW$.

\begin{thm}[Gaifman normal form for $\FOWS$]\label{thm:GaifmanForFOW}
Every $\FOWS$-formula $\phi$ is equivalent to an $\FOWS$-formula
$\gamma$ in Gaifman normal form with $\free(\gamma)=\free(\phi)$.
Furthermore, there is an algorithm that computes $\gamma$ upon input
of $\phi$.
\end{thm}

The proof proceeds similarly
as Gaifman's original proof for first-order logic $\FO$
(\cite{Gai82}, see also~\cite[Sect.~4.1]{Gro08}),
but since subformulas are from $\FOWS$, we use Corollary~\ref{cor:FV} instead of
Feferman-Vaught decompositions for $\FO$
(cf.~\cite[Lemma~2.3]{Gro08}).
Furthermore, for formulas built according to
rule~\eqref{item:finitegroup}$_1$,
we proceed in a similar way as for the
modulo-counting quantifiers in the Gaifman normal
construction of~\cite{KuskeSchweikardt_Gaifman}.

\bigskip

\noindent
The remainder of Section~\ref{sec:Gaifman} is devoted to the proof of
Theorem~\ref{thm:GaifmanForFOW}.

\begin{proof}[Proof of Theorem~\ref{thm:GaifmanForFOW}] \ \\
The proof proceeds by induction on the construction of $\phi$.
The cases where formulas are built according to the rules
\eqref{item:atomic}, \eqref{item:wsimple}, \eqref{item:bool} of
Definition~\ref{def:WA} are trivial.
A formula $\phi$ that is built according to rule \eqref{item:Q}$_1$
is of the form
$\P(t_1,\ldots,t_m)$, where $\P\in\Ps$ and $t_1,\ldots,t_m$
are $\SC$-terms built using the rules
\eqref{item:constterm}--\eqref{item:plustimesterm} ---
thus, $\phi$ is $0$-local.

If $\phi$ is of the form $\exists y\,\phi'$, we can argue in the same
way as in Gaifman's original proof for first-order logic
(\cite{Gai82}, see also \cite[Sect.~4.1]{Gro08}),
but since $\phi'$ is from $\FOWS$, we use Corollary~\ref{cor:FV} instead of
Feferman-Vaught decompositions for first-order logic (cf.\
\cite[Lemma~2.3]{Gro08}).

For formulas built according to rule \eqref{item:finitegroup}$_1$ of
Definition~\ref{def:FOW}, we proceed in a similar way as for the
modulo-counting quantifiers in the Gaifman normal
construction of \cite{KuskeSchweikardt_Gaifman}. Let $\phi$ be of the
form $\big(s=\sum\weight(\ov{y}).\phi'(\ov{x},\ov{y})\big)$, for a
tuple of variables $\ov{y}=(y_1,\ldots,y_\ell)$ and a weight symbol
$\weight\in\Weights$ whose type $S\deff\wtype(\weight)$ is finite, and
let $\ov{x}=(x_1,\ldots,x_k)$ be the free variables of $\phi$ (note
that $k$ might be 0).
By the induction hypothesis, we can transform $\phi'$ into an
equivalent formula in Gaifman normal form, and we can assume w.l.o.g.\
that this formula is of the form
\ $
  \Oder_{j=1}^n \big( \chi_j \und \lambda_j(\ov{x},\ov{y}) \big)\,,
$ \
where each $\chi_j$ is an $\FOWS$-sentence in Gaifman normal form and
each $\lambda_j(\ov{x},\ov{y})$ is $r$-local, for some $r\in\NN$.
For every $J\subseteq [n]$, let
\[
  \chi_J\ \deff\ \Und_{j\in J}\chi_j\und\Und_{j\in [n]\setminus
    J}\nicht\chi_j
  \qquad\text{and}\qquad
  \lambda_J(\ov{x},\ov{y}) \ \deff \ \Oder_{j\in J}\lambda_j(\ov{x},\ov{y})\,.
\]
Clearly,
$
  \Oder_{j=1}^n \big( \chi_j \und \lambda_j(\ov{x},\ov{y}) \big)
$
is equivalent to
$
  \Oder_{\emptyset\neq J\subseteq [n]} \big( \chi_J \und \lambda_J(\ov{x},\ov{y}) \big)\,
$,
the $(\chi_J)_{J\subseteq [n]}$ are mutually exclusive sentences in
Gaifman normal form, and $\chi_J(\ov{x},\ov{y})$ is $r$-local.
Let
\[
 \tilde{\phi} \ \deff \
 \Oder_{\emptyset\neq J\subseteq [n]} \Big(
   \chi_J \ \und \
   \big(
     s = \sum \weight(\ov{y}).\lambda_J(\ov{x},\ov{y})
   \big)
 \Big).
\]
The following is straightforward to prove.
\begin{claim}\label{claim:Gaifman1}
 If $s\neq\nullS$, then $\phi$ is equivalent to $\tilde{\phi}$.
 If $s=\nullS$, then $\phi$ is equivalent to $(\tilde{\phi}\oder\chi_{\emptyset})$.
\end{claim}
To complete the proof of Theorem~\ref{thm:GaifmanForFOW}, it suffices
to consider an arbitrary non-empty $J\subseteq[n]$ and the $r$-local
formula $\lambda(\ov{x},\ov{y})\deff\lambda_J(\ov{x},\ov{y})$ and show
how to transform the formula
$\psi(\ov{x})\deff \big(s=\sum\weight(\ov{y}).\lambda(\ov{x},\ov{y})\big)$ into an equivalent
formula in Gaifman normal form. If $k=0$, we are done since $\psi$ is
a local aggregation sentence in $\FOWS$. If $k>0$, we proceed as follows.
Let $r'\deff 2r{+}1$ and $I\deff\setc{(i_1,i_2)\in S\times
  S}{i_1\plusS i_2 = s}$.
Then, $\psi(\ov{x})$ is equivalent to
\ $\Oder_{(i_1,i_2)\in I} \big( \psi'_{i_1} \und \psi''_{i_2}\big)$,
where
\begin{align*}
\psi'_{i_1}(\ov{x}) \ &\deff \
\big(\;
  i_1 = \sum \weight(\ov{y}).\Big(
    \lambda(\ov{x},\ov{y}) \;\und \;
    \nicht\, \big( \Und_{i=1}^k \Und_{j=1}^\ell \dist(x_i,y_j)>r' \big)
  \Big)
\;\big),\\
\psi''_{i_2}(\ov{x}) \ &\deff \
\big(\;
  i_2 = \sum \weight(\ov{y}).\Big(
    \lambda(\ov{x},\ov{y}) \;\und \;
    \ \big( \Und_{i=1}^k \Und_{j=1}^\ell \dist(x_i,y_j)>r' \big)
  \Big)
\;\big).
\end{align*}
Note that the formula $\psi'_{i_1}(\ov{x})$ is local (namely,
$(r'{+}1{+}r)$-local; this is because tuples $\ov{a}$ in a
$(\sigma,\Weights)$-structure $\A$ with $\weight^\A(\ov{a})\neq\nullS$
must form a clique in the Gaifman graph of $\A$).

It remains to transform $\psi''_{i_2}$ into an equivalent formula in
Gaifman normal form. To achieve this, we use Corollary~\ref{cor:FV} to
obtain a finite, non-empty set $\Delta$ of pairs
$\big(\alpha(\ov{x}),\beta(\ov{y})\big)$ of $r$-local $\FOWS$-formulas such that
$\Big(
    \lambda(\ov{x},\ov{y}) \,\und \,
    \big( \Und_{i=1}^k \Und_{j=1}^\ell \dist(x_i,y_j)>r' \big)
  \Big)$
is equivalent to
$\Big(
    \Oder_{(\alpha,\beta)\in\Delta}(\alpha(\ov{x}) \und \beta(\ov{y})) \,\und \,
    \big( \Und_{i=1}^k \Und_{j=1}^\ell \dist(x_i,y_j)>r' \big)
  \Big)$.

W.l.o.g., we can assume that \emph{the $\alpha$s in $\Delta$ are
  mutually exclusive}, i.e.\ for any two distinct $(\alpha,\beta)$ and
$(\alpha',\beta')$ in $\Delta$, the formula $(\alpha\und\alpha')$ is
unsatisfiable.
Thus, $\psi''_{i_2}(\ov{x})$ is equivalent to the formula
$\big(
 i_2 = \sum \weight(\ov{y}).\big(
 \Oder_{(\alpha,\beta)\in\Delta}(\alpha(\ov{x}) \und \beta(\ov{y})) \und
    \big( \Und_{i=1}^k \Und_{j=1}^\ell \dist(x_i,y_j)>r' \big)
\big)\big)$.
Let
\[
\tilde{\psi}_{i_2}(\ov{x}) \ \deff\
\Oder_{(\alpha,\beta)\in\Delta}\Big(\;
  \alpha(\ov{x}) \,\und \,
  \Big(
     i_2 = \sum \weight(\ov{y}). \big( \,
       \beta(\ov{y}) \,\und \, \big( \Und_{i=1}^k \Und_{j=1}^\ell
       \dist(x_i,y_j)>r' \big) \,\big)
  \Big)
\;\Big)\,.
\]
Let $\mathbb{A}\deff\setc{\alpha}{\text{there exists }\beta\text{ such that
  }(\alpha,\beta)\in \Delta}$.
The following is straightforward to prove:
\begin{claim}\label{claim:Gaifman2}
If $i_2\neq\nullS$, then $\psi''_{i_2}(\ov{x})$ is equivalent to
$\tilde{\psi}_{i_2}(\ov{x})$.
If $i_2=\nullS$, then $\psi''_{i_2}(\ov{x})$ is equivalent to
\,$\big(\, \tilde{\psi}_{i_2}(\ov{x}) \,\oder \,
\Und_{\alpha\in\mathbb{A}}\nicht\alpha(\ov{x})\,\big)$.
\end{claim}

\noindent
To complete the proof of Theorem~\ref{thm:GaifmanForFOW}, it suffices
to consider an arbitrary $r$-local formula $\beta(\ov{y})$ and
transform the formula
\ $
  \mu(\ov{x}) \deff
  \Big(
    i_2 = \sum \weight(\ov{y}). \big( \,
       \beta(\ov{y}) \,\und \, \big( \Und_{i=1}^k \Und_{j=1}^\ell
       \dist(x_i,y_j)>r' \big) \,\big)
  \Big)
$ \
into an equivalent $\FOWS$-formula in Gaifman normal form.
This is not difficult: let $J\deff\setc{(j_1,j_2)\in S\times
  S}{j_1\minusS j_2 = i_2}$. Then, $\mu(\ov{x})$ is equivalent to
$
  \Oder_{(j_1,j_2)\in J} \big( \mu'_{j_1} \und \mu''_{j_2}(\ov{x}) \big)
$, where
\begin{eqnarray*}
  \mu'_{j_1} & \deff &
  \big(\, j_1 = \sum \weight(\ov{y}).\beta(\ov{y}) \,\big)
\\
  \mu''_{j_2}(\ov{x}) & \deff&
    \big(\,
      j_2 = \sum \weight(\ov{y}).\big(\,
        \beta(\ov{y})\und
        \nicht\, \big( \Und_{i=1}^k \Und_{j=1}^\ell \dist(x_i,y_j)>r' \big)
      \,\big)
    \,\big)\,.
\end{eqnarray*}
Now, $\mu'_{j_1}$ is a local aggregation sentence in $\FOWS$, and
$\mu''_{j_2}$ is local (namely, $(r'{+}1{+}r)$-local; this is because tuples $\ov{a}$ in a
$(\sigma,\Weights)$-structure $\A$ with $\weight^\A(\ov{a})\neq\nullS$
must form a clique in the Gaifman graph of $\A$).
This completes the proof of Theorem~\ref{thm:GaifmanForFOW}.
\end{proof}

\subsection{Localisation Theorem for $\WAun$}\label{sec:Decomposition}
Our next main theorem provides a locality result for the logic
$\WAun$, which is a logic substantially more expressive than $\FOW$.

\begin{thm}[Localisation Theorem for $\WAun$]\label{thm:AlgorithmicDecomp}
For every $\WAunPS$-formula $\phi(x_1,\ldots,x_k)$ (with $k\geq 0$), there is an
extension $\sigma_\phi$ of $\sigma$ with relation symbols of arity
$\leq 1$, and a $\FOWSR{\PP}{\sigma_\phi,\SC,\Weights}$-formula
$\phi'(x_1,\ldots,x_k)$ that is a Boolean combination of local
formulas and statements of the form $R()$ where $R\in\sigma_\phi$ has
arity $0$, for which the following is true:
there is an algorithm\footnote{with $\PP$- and $\SC$-oracles, so that
  operations $\plusS,\malS$ for $S\in\SC$ and checking if a tuple
  belongs to $\sem{\P}$ for $\P\in\PP$ can be done in constant time} that, upon input of a
$(\sigma,\Weights)$-structure $\A$, computes in
time $|A|{\cdot} d^{\bigOh(1)}$, where $d$ is the degree of $\A$, a
$\sigma_\phi$-expansion $\A^{\phi}$ of $\A$ such that
for all $\ov{a}\in A^k$ it holds that
$\A^\phi\models\phi'[\ov{a}]$ $\iff$ $\A\models\phi[\ov{a}]$.
\end{thm}

\noindent
The remainder of \cref{sec:Decomposition} is devoted to the proof of
Theorem~\ref{thm:AlgorithmicDecomp}.
Our approach is to
decompose $\WAun$-expressions
into simpler
expressions that can be
evaluated in a structure $\A$ by exploring for each element $a$ in
the universe of $\A$ only a local neighbourhood around $a$. This is
achieved by a
decomposition theorem (Theorem~\ref{thm:normalformWAun}), which
is a generalisation of the decomposition for $\FOunCP$ provided in
\cite[Theorem~6.6]{GroheSchweikardt_FOunC}.

\subsubsection{Connected local terms}

The following well-known lemma summarises easy facts concerning neighbourhoods.

\begin{lem}\label{lem:basic_facts}
  Let $\A$ be a  $(\sigma,\Weights)$-structure,
  $r\geq 0$, $k\geq 1$, and $\ov{a}=(a_1,\ldots,a_k)\in A^k$.
\\
    $\NrA{a_1,a_2}$ is connected $\iff$
    $\dist^\A(a_1,a_2)\leq 2r{+}1$.
\\
    If $\NrA{\ov{a}}$ is connected, then
    $\nrA{\ov{a}}\subseteq \neighb{r+(k-1)(2r+1)}{\A}{a_i}$,
    for each $i\in[k]$.
\end{lem}

For every $k\in\NNpos$, we let $\Graphclass_k$ be the set of all
undirected graphs $G$ with vertex set $[k]$.
For a graph $G \in\Graphclass_k$, a number $r\in\NN$, and a
tuple $\ov{y}=(y_1,\ldots,y_k)$ of $k$ pairwise distinct variables,
we consider the formula
\begin{equation*}\label{eq:delta-formula}
\displaystyle
  \delta^\sigma_{G,r}(\ov{y})
  \quad \deff \quad
  \ \
  \Und_{\set{i,j}\in E(G)}\!\!\!\!\! \dist^{\sigma}(y_i,y_j)\,{\leq}\, r
  \ \ \und \!\!\!\!
  \Und_{\set{i,j}\not\in E(G)}\!\!\!\!\dist^{\sigma}(y_i,y_j)\,{>}\,r\,.
\end{equation*}
Note that $\A\models\delta^\sigma_{G,2r+1}[\ov{a}]$ means that the
connected components of the
$r$-neighbourhood $\NrA{\ov{a}}$ correspond to the connected
components of $G$.
Clearly, the formula
$\delta^{\sigma}_{G,2r+1}(\ov{y})$ is $r$-local around its free variables $\ov{y}$.

The main ingredient of our decomposition of $\WAunPS$-expressions
are
the connected local terms (\emph{cl-terms}, for short),
defined as follows.

\begin{defn}[cl-Terms]\label{def:cl-term}
Let $r\in\NN$ and $k\in\NNpos$.
\\
 A \emph{basic cl-term (of radius $r$ and width $k$)} is an $\SC$-term
 of the form
 \[
   \Count{\pweight}{
     \big(\,\psi(y_1,\ldots,y_k)\;\und\;\delta^\sigma_{G,2r+1}(y_1,\ldots,y_k)\,\big) }
 \]
 where $\varsof(\pweight)\subseteq\set{y_1,\ldots,y_k}$,
 $\ov{y}=(y_1,\ldots,y_k)$ is a tuple of $k$ pairwise distinct variables,
 $\psi(y_1,\ldots,y_k)$ is an $\FOWS$-formula that is $r$-local around
 $\ov{y}$, and $G\in\Graphclass_k$ is \emph{connected}.
 A \emph{cl-term (of radius $\leq r$ and width $\leq k$)} is built from basic cl-terms (of
 radius $\leq r$ and width $\leq k$) by
 using rules \eqref{item:constterm}--\eqref{item:plustimesterm} of Definition~\ref{def:WA}.
\end{defn}

Note that cl-terms are ``easy'' with
respect to query evaluation in the following sense.

\begin{lem}\label{lem:clterm-evaluation}
For every fixed cl-term $t(z_1,\ldots,z_\ell)$ (with $\ell\geq 0$),
there is an algorithm which, upon input of a
$(\sigma,\Weights)$-structure $\A$, can compute, within precomputation time
$|A|\cdot d^{\bigOh(1)}$ where $d$ is the degree of $\A$, a data
structure that, whenever given a tuple
$(a_1,\ldots,a_\ell)\in A^\ell$, returns the value
$t^\A[a_1,\ldots,a_\ell]$ in constant time.
\end{lem}
\begin{proof}
It suffices to prove the lemma for basic cl-terms. The statement for
general cl-terms then follows by induction.
Consider a basic cl-term $u(z_1,\ldots,z_\ell)$ of the form
\ $
  \Count{\pweight}{\big(\,\psi(y_1,\ldots,y_k)\und
  \delta^\sigma_{G,2r+1}(y_1,\ldots,y_k)\,\big)}
$.
Recall from Definition~\ref{def:cl-term} that $G$ is a \emph{connected} graph and
$\set{z_1,\ldots,z_\ell}\subseteq\set{y_1,\ldots,y_k}$. Let $S\in\SC$
be the type of the $\Weights$-product $\pweight$.
We can assume w.l.o.g.\ that $(z_1,\ldots,z_\ell)=(y_1,\ldots,y_\ell)$.
Consequently, $\varsof(\pweight)=\set{y_{\ell+1},\ldots,y_k}$.

Given a $(\sigma,\Weights)$-structure $\A$ and an element $c_1\in A$,
we can explore the $R$-neighbourhood of $c_1$ for
$R\deff r+(k{-}1)(2r{+}1)$ (cf.\ Lemma~\ref{lem:basic_facts}) and thereby compute
the set $M_{c_1}$ of all $\ov{a}=(a_1,\ldots,a_k)\in A^k$ with $a_1=c_1$ such that
$(\A,\ov{a})\models(\psi\und\delta^\sigma_{G,2r+1})$.
For each such tuple $\ov{a}$, we compute and store the value $v_{\ov{a}}\deff \pweight^{\A}[a_{\ell+1},\ldots,a_k]\in S$.
Then, we group the tuples in $M_{c_1}$ by their prefix $(a_1,\ldots,a_\ell)$ of length $\ell$, and
for each group, we compute the $\plusS$-sum $s_{c_1,(a_1,\ldots,a_\ell)}$ of the values $v_{\ov{a}}$ of all tuples $\ov{a}\in M_{c_1}$ that have the same prefix $(a_1,\ldots,a_\ell)$.

In case that $\ell=0$, $u$ is a ground term and we have
$u^\A = \sum_S \setc{s_{c_1,()}}{c_1\in A}$.
In case that $\ell\geq 1$, whenever given an arbitrary tuple $(a_1,\ldots,a_\ell)\in A^\ell$, we can determine $u^\A[a_1,\ldots,a_\ell]$ as follows: let $c_1\deff a_1$, if $M_{c_1}$ contains a tuple with prefix $(a_1,\ldots,a_\ell)$ then
$u^\A[a_1,\ldots,a_\ell]=s_{c_1,(a_1,\ldots,a_\ell)}$, and
otherwise $u^\A[a_1,\ldots,a_\ell]=\nullS$.

Thus, upon input of a $(\sigma,\Weights)$-structure $\A$, we can, within
precomputation time $|A|\cdot d^{\bigOh(1)}$ where $d$ is the degree
of $\A$, compute a data structure which, whenever given a tuple
$(a_1,\ldots,a_\ell)\in A^\ell$, returns the value
$u^\A[a_1,\ldots,a_\ell]$ in constant time.
\end{proof}

Our decomposition of $\WAunPS$-expressions
proceeds by induction on the construction of the input expression.
The main technical tool for the construction is
the following lemma.

\begin{lem}\label{lem:normalform:terms}
Let $r\geq 0$,
$k\geq 1$, and let $\ov{y}=(y_1,\ldots,y_k)$ be a tuple of
$k$ pairwise distinct variables.
Let $\psi(\ov{y})$ be an  $\FOWS$-formula
that is $r$-local,
and consider an $\SC$-term $u(z_1,\ldots,z_m)$ of the form
\ $\Count{\pweight}{\psi(y_1,\ldots,y_k)}$,
where $\pweight$ is a $\Weights$-product, $m\geq 0$, and
$\set{z_1,\ldots,z_m}\subseteq\set{y_1,\ldots,y_k}$.
There exists a cl-term $\hat{u}(z_1,\ldots,z_m)$
of radius $\leq r$ and width $\leq k$,
such that
$\hat{u}^{\A}[\ov{a}] = u^{\A}[\ov{a}]$ holds
for every $(\sigma,\Weights)$-structure $\A$ and every $\ov{a}\in A^m$.
Furthermore, there is an algorithm which, upon input of $r$ and $u$,
constructs $\hat{u}$.
\end{lem}

\begin{proof}
For a $(\sigma,\Weights)$-structure $\A$ and a formula $\vartheta(\ov{y})$, we
consider the set
\[
 S_\vartheta^\A \quad \deff\quad
 \setc{\ \ov{a}=(a_1,\ldots,a_k)\in A^k \ }{\ \A\models\vartheta[\ov{a}]\ }\,.
\]
Note that for every graph $G\in\Graphclass_k$, the formula
\[
  \psi_G(\ov{y}) \quad \deff \quad
  \psi(\ov{y})\;\und\; \delta^{\sigma}_{G,2r+1}(\ov{y})
\]
is $r$-local around $\ov{y}$. Furthermore, for every
$(\sigma,\Weights)$-structure $\A$, the set $S_\psi^{\A}$ is the disjoint union of
the sets $S_{\psi_G}^{\A}$ for all $G\in \Graphclass_k$.
Therefore, $u$ is equivalent to the $\plus$-sum, over all
$G\in\Graphclass_k$, of the $\SC$-terms
\ $
  u^\psi_G
  \deff
 \Count{\pweight}{\psi_G(y_1,\ldots,y_k)}.
$
To complete the proof of Lemma~\ref{lem:normalform:terms}, it
therefore suffices to show that, for every $G\in\Graphclass_k$, the
$\SC$-term $u^\psi_G$ is
equivalent to a cl-term of radius $r$.
We prove this by an
induction on the number of connected components of $G$. Precisely, we
show that the following statement $(*)_c$ is true for every $c\in\NNpos$.

\begin{enumerate}[$(*)_c$:]
\item[$(*)_c$:]
  For every $k\geq c$, for every tuple $\ov{y}=(y_1,\ldots,y_k)$ of
  $k$ pairwise distinct variables, for every $r\geq 0$, for every
  $\FOWS$-formula $\psi(\ov{y})$ that is $r$-local around
  $\ov{y}$, for every $\Weights$-product $\pweight$ with
  $\varsof(\pweight)\subseteq\set{y_1,\ldots,y_k}$,  and for every
  graph $G\in\Graphclass_k$ that has at most
  $c$ connected components, the $\SC$-term $u^\psi_G\deff
  \Count{\pweight}{\psi_G(y_1,\ldots,y_k)}$
  is equivalent to a cl-term of radius $r$.
\end{enumerate}

The induction base for $c=1$ is trivial: it involves
only \emph{connected} graphs $G$, for which by Definition~\ref{def:cl-term},
$u^\psi_G$ is a basic cl-term.

For the induction step from $c$ to $c{+}1$, consider a $k\geq
c{+}1$ and a graph $G=(V,E)\in\Graphclass_k$ that has $c{+}1$ connected
components.
Let $V'$ be the set of all vertices of $V$ that are connected to the
vertex $1$, and let $V''\deff V\setminus V'$.

Let $G'\deff\inducedSubStr{G}{V'}$ and $G''\deff\inducedSubStr{G}{V''}$
be the induced subgraphs of $G$ on $V'$ and $V''$, respectively.
Clearly, $G$ is the disjoint union of $G'$ and $G''$, $G'$ is connected, and $G''$ has $c$ connected components.

To keep notation simple, we assume (without loss of generality)
that $V'=\set{1,\ldots,\ell}$ and $V''=\set{\ell{+}1,\ldots,k}$ for
an $\ell$ with $1\leq \ell<k$.
For a tuple $\ov{v}=(v_1,\ldots,v_k)$, we
let $\ov{v}{}'\deff (v_1,\ldots,v_\ell)$ and $\ov{v}{}''\deff
(v_{\ell+1},\ldots,v_{k})$.

Now consider a number $r\geq 0$ and the formula
$\delta^{\sigma}_{G,2r+1}(\ov{y})$ for $\ov{y}=(y_1,\ldots,y_k)$.
For every $\sigma$-structure $\A$
and every tuple $\ov{a}=(a_1,\ldots,a_k)\in A^k$ with
$\A\models\delta^{\sigma}_{G,2r+1}[\ov{a}]$,
the $r$-neighbourhood $\NrA{\ov{a}}$ is the disjoint union of the
$r$-neighbourhoods $\NrA{\ov{a}{}'}$ and $\NrA{\ov{a}{}''}$.

Let $\psi(\ov{y})$ be an $\FOWS$-formula that is $r$-local.
By using Corollary~\ref{cor:FV},
we can compute a decomposition of $\psi(\ov{y})$ into a formula
$\hat{\psi}(\ov{y})$
of the form
\[
  \Oder_{i\in I}\ \ \Big(\
   \psi_i{}'(\ov{y}{}') \ \und \
   \psi_i{}''(\ov{y}{}'')
  \ \Big)\,,
\]
where $I$ is a finite non-empty set,
each $\psi_i{}'(\ov{y}{}')$ is an $\FOWS$-formula that is
$r$-local around $\ov{y}{}'$,
each $\psi_i{}''(\ov{y}{}'')$ is an $\FOWS$-formula that is
$r$-local around $\ov{y}{}''$, and
for every $(\sigma,\Weights)$-structure
$\A$ and every $\ov{a}\in A^k$ with
$\A\models\delta^{\sigma}_{G,2r+1}[\ov{a}]$, the following is true:
 there exists at most one $i\in I$ such that
 $(\A,\ov{a}) \models
  \big(\,\psi_i{}'(\ov{y}{}') \und \psi_i{}''(\ov{y}{}'')\,\big)$, \ and
 $\A\models\psi[\ov{a}]  \iff
  \A\models\hat{\psi}[\ov{a}]$\,.
This implies that the set $S^{\A}_{\psi_G}$ is the
disjoint union
of the sets
$S^{\A}_{(\psi_i{}'\und\psi_i{}''\und\delta^{\sigma}_{G,r})}$ for all
$i\in I$.

Now let $\pweight$ be an arbitrary $\Weights$-product with
$\varsof(\pweight)\subseteq \set{y_1,\ldots,y_k}$, and consider the
$\SC$-term
\ $
  u^\psi_G \deff
  \Count{\pweight}{\psi_G(y_1,\ldots,y_k)}$.
From the above reasoning, it follows that $u^\psi_G$ is equivalent to
the $\plus$-sum, over all $i\in I$, of the $\SC$-terms
\[
 u_G^{\psi,i} \ \deff \
  \Count{\pweight}{\big(\,\psi_i{}'(\ov{y}')\,\und\,\psi_i{}''(\ov{y}'')\,\und\,\delta^{\sigma}_{G,2r+1}(\ov{y})\,\big)}\,.
\]
To complete the proof, it suffices to show that
$u_G^{\psi,i}$
is equivalent to a cl-term of radius $r$.

By the definition of the formula $\delta^{\sigma}_{G,2r+1}(\ov{y})$, we
obtain that the formula
\
$\psi_i{}'(\ov{y}')\,\und\,\psi_i{}''(\ov{y}'')\,\und\,\delta^{\sigma}_{G,2r+1}(\ov{y})$
\
is equivalent to the formula
\begin{equation}\label{eq:theta-formulas}
\underbrace{ \Big(\; \psi_i{}'(\ov{y}')\,\und\,\delta^{\sigma}_{G',2r+1}(\ov{y}')\;
\Big)}_{\textstyle =:\ \vartheta'(\ov{y}')}
 \; \und \;
\underbrace{ \Big(\; \psi_i{}''(\ov{y}'')\,\und\,\delta^{\sigma}_{G'',2r+1}(\ov{y}'')\;
\Big)}_{\textstyle =:\ \vartheta''(\ov{y}'')}
\ \und\!\!
 \Und_{j'\in V'\atop j''\in V''} \!\!\!\!\dist^{\sigma}(y_{j'},y_{j''})>2r{+}1 \,.
\end{equation}
Therefore, for every $(\sigma,\Weights)$-structure $\A$, we have
\[
 S^{\A}_{\psi'_i\und\psi''_i\und\delta^{\sigma}_{G,r}}
 \quad = \quad
 \big( S^{\A}_{\vartheta'}\times S^{\A}_{\vartheta''} \big) \
 \setminus\ T^{\A}\,,
 \qquad\quad \text{for}
\]
\[
  T^{\A} \quad \deff \quad
  \bigsetc{\
   \ov{a}\in A^k \ }{\
   \A\models\vartheta'[\ov{a}'], \
   \A\models\vartheta''[\ov{a}''],\
   (\A,\ov{a})\not\models  \Und_{j'\in V'\atop j''\in V''}
   \dist^{\sigma}(y_{j'},y_{j''})>2r{+}1
  \ }\,.
\]
Let $\GraphclassH$ be the set of all graphs
$H\in\Graphclass_k$ with $H\neq G$, but
$\inducedSubStr{H}{V'}=G'$ and $\inducedSubStr{H}{V''}=G''$.
Clearly, every $H\in\GraphclassH$
has at most $c$ connected components. Furthermore,
for every $(\sigma,\Weights)$-structure $\A$, the set $T^{\A}$ is the disjoint
union over all $H\in\GraphclassH$ of the sets
\[
  T^{\A}_{H}
  \quad\deff\quad
  \bigsetc{\
   \ov{a}\in A^k \ \;}{\ \;
   \A\models\vartheta'[\ov{a}'], \ \
   \A\models\vartheta''[\ov{a}''],\ \
   \A\models\delta^{\sigma}_{H,2r+1}[\ov{a}]
  \ }\,.
\]

Now let us have a closer look at the $\Weights$-product $\pweight$
used in the $\SC$-term $u_G^{\psi,i}$.
We let $Y'\deff\varsof(\pweight)\cap \set{y_1,\ldots,y_\ell}$ and
$Y''\deff\varsof(\pweight)\cap\set{y_{\ell+1},\ldots,y_k}$.

\smallskip

\emph{Case~1:} $\pweight$ contains a factor $\weight(\ov{z})$ for some
$\weight\in\Weights$ and a tuple $\ov{z}$ that contains variables
from both $Y'$ and $Y''$.
Then, for every
$(\sigma,\Weights)$-structure $\A$ and every $\ov{a}\in S^\A_{\psi'_i
  \und\psi''_i\und\delta^{\sigma}_{G,r}}$, we know that
$\sem{p}^{(\A,\ov{a})}=\nullS$, where $S\in\SC$ is the type of
$\pweight$.
Hence, $u_G^{\psi,i}$ is equivalent to
the $\SC$-term $\nullS$, and we are done.

\smallskip

\emph{Case~2:} If case~1 does not apply, then $\pweight$ is
of the form
$\pweight'_1\mal\pweight''_1\mal\cdots\mal\pweight'_j\mal\pweight''_j$,
where $j\geq 1$ and
$\varsof(\pweight'_i)\subseteq Y'$ and $\varsof(\pweight''_i)\subseteq
Y''$ for each $i\in[j]$.

Now let us consider an arbitrary $(\sigma,\Weights)$-structure $\A$
and fix an assignment $\beta$ to the free variables of
$u_G^{\psi,i}$. Evaluating $u_G^{\psi,i}$ in $(\A,\beta)$ means
computing the value
\[
s^{(\A,\beta)} \ \deff \
  {\textstyle\sum_S}\ \bigsetc{\;\sem{\pweight}^{(\A,\ov{a})}\;}{\;\ov{a}\in
  S^\A_{\psi'_i\und\psi''_i\und \delta^{\sigma}_{G,r}} \text{ such
    that $\ov{a}$ agrees with $\beta$ on $\free(u_G^{\psi,i})$  }}\,.
\]
We already know that $S_{\psi'_i\und\psi''_i\und\delta^{\sigma}_{G,r}}
= (S^\A_{\vartheta'} \times S^\A_{\vartheta''}) \setminus
(\bigcup_{H\in\GraphclassH} T^\A_H)$, where the sets $T^\A_H$ for
$H\in\GraphclassH$ are pairwise disjoint and contained in
$S^\A_{\vartheta'}\times S^\A_{\vartheta''}$.
Therefore, $s^{(\A,\beta)} =$
\[
\begin{array}{cl}
&
  {\textstyle\sum_S}\ \setc{\;\sem{\pweight}^{(\A,\ov{a})}\;}{\;\ov{a}\in
  S^\A_{\vartheta'}\times S^\A_{\vartheta''} \text{ such
    that $\ov{a}$ agrees with $\beta$ on $\free(u_G^{\psi,i})$  }}
\smallskip\\
\minusS\ \ \big( &
 \sum_{H\in\GraphclassH} \
\sum_S \setc{\;\sem{\pweight}^{(\A,\ov{a})}\;}{\;\ov{a}\in
   T^\A_H \text{ such
    that $\ov{a}$ agrees with $\beta$ on $\free(u_G^{\psi,i})$  }} \ \
\big)\,.
\end{array}
\]
Furthermore, since
$\pweight=\pweight'_1\mal\pweight''_1\mal\cdots\mal\pweight'_j\mal\pweight''_j$,
we obtain that\\
 $  {\textstyle\sum_S}\ \setc{\;\sem{\pweight}^{(\A,\ov{a})}\;}{\;\ov{a}\in
  S^\A_{\vartheta'}\times S^\A_{\vartheta''} \text{ such
    that $\ov{a}$ agrees with $\beta$ on $\free(u_G^{\psi,i})$  }} \ \
= $
\[
\begin{array}{r}
{\displaystyle \prod_{i=1}^j}\ \Big(\
  {\textstyle\sum_S}\; \setc{\;\sem{\pweight'_i}^{(\A,\ov{a}')}\;}{\;\ov{a}'\in
  S^\A_{\vartheta'} \text{ such
    that $\ov{a}'$ agrees with $\beta$ on
    $\free(\vartheta')\setminus\varsof(\pweight')$  }}
\hspace*{8ex}
\\
\ \malS \ \
  {\textstyle\sum_S}\; \setc{\;\sem{\pweight''_i}^{(\A,\ov{a}'')}\;}{\;\ov{a}''\in
  S^\A_{\vartheta''} \text{ such
    that $\ov{a}''$ agrees with $\beta$ on
    $\free(\vartheta'')\setminus\varsof(\pweight'')$  }}\ \Big)\,.
\end{array}
\]
Therefore, $u_G^{\psi,i}$ is equivalent to
\[
 {\displaystyle \prod_{i=1}^j}\; \Big(\;
  \underbrace{\Big( \Count{\pweight'_i}{\vartheta'(\ov{y}')}
    \Big)}_{\textstyle =: \ t'_i}
  \cdot
  \underbrace{\Big( \Count{\pweight''_i}{\vartheta''(\ov{y}'')}
    \Big)}_{\textstyle =:\ t''_i} \;\Big)
  \ - \
  \sum_{H\in\GraphclassH}
  \underbrace{\Count{\pweight}{\big(\;\vartheta'(\ov{y}')\und\vartheta''(\ov{y}'')\und\delta^{\sigma}_{H,2r+1}(\ov{y})
      \;\big)}}_{\textstyle =: t_H}\,.
\]
By the induction hypothesis $(*)_c$, each of the terms $t'_i$, $t''_i$, and $t_H$
is equivalent to a cl-term of radius $r$.
Hence, also $u_G^{\psi,i}$ is equivalent to a cl-term of radius $r$.
This completes the proof of Lemma~\ref{lem:normalform:terms}.
\end{proof}

As an easy consequence of Lemma~\ref{lem:normalform:terms}, we obtain

\begin{lem}\label{lem:normalform:terms2}
 Let $s\geq 0$ and let
  $\chi_1, \ldots, \chi_s$ be arbitrary
  sentences that can be evaluated in $(\sigma,\Weights)$-structures.\footnote{We do not restrict
    attention to $\FOWS$-sentences here---the $\chi_j$s may
    be sentences of any logic, e.g., $\WAPS$.}
  Let $r\geq 0$, $k\geq 1$, and let
  $\ov{y}=(y_1,\ldots,y_k)$ be a tuple of $k$ pairwise
  distinct variables. Let
  $\phi(\ov{y})$ be a Boolean combination of the
  sentences $\chi_1,\ldots,\chi_s$ and
  of $\FOWS$-formulas that are $r$-local around their free
  variables $\ov{y}$.
  Consider an $\SC$-term $u(z_1,\ldots,z_m)$ of the form
  $
    \Count{\pweight}{\phi(y_1,\ldots,y_k)}$,
  where
  $\pweight$ is a $\Weights$-product, $m\geq 0$, and
  $\set{z_1,\ldots,z_m}\subseteq\set{y_1,\ldots,y_k}$.
  For every $J\subseteq [s]$, there is a cl-term $\hat{u}_J$
  (of radius $\leq r$ and width $\leq k$) such
  that for every $(\sigma,\Weights)$-structure
  $\A$, there is exactly one set $J\subseteq [s]$ such that
  \[
    \A \ \ \models \ \
    \chi_J \ \ \deff \ \ \Und_{j\in J}\chi_j \ \und \Und_{j\in
      [s]\setminus J} \nicht\,\chi_j\,,
  \]
  and for this set $J$, we have
  \  $\hat{u}_J^\A[\ov{a}] = u^\A[\ov{a}]$ \ for every $\ov{a}\in A^m$.
\\
  Furthermore, there is an algorithm
  which upon input of $r$, $u$, and $J$ constructs $\hat{u}_J$.
\end{lem}
\begin{proof}
We can assume w.l.o.g.\ that $\phi(\ov{y})$ is of the form
\[
  \Oder_{J\subseteq [s]} \big(\;
   \chi_J \ \und \ \psi_J(\ov{y})
  \;\big)
\]
where, for each $J\subseteq[s]$, $\psi_J(\ov{y})$ is an
$\FOWS$-formula that is $r$-local around its free variables
$\ov{y}$.

For every $J\subseteq [s]$ let $\hat{u}_J$ be the
cl-term obtained by Lemma~\ref{lem:normalform:terms} for the term
$u_J\deff\Count{\pweight}{\psi_J(\ov{y})}$.
Recall that $u=\Count{\pweight}{\phi(\ov{y})}$.

Now consider an arbitrary $J\subseteq [s]$ and a $\sigma$-structure
$\A$ with $\A\models\chi_J$. Clearly, for every $\ov{a}\in A^m$ we have
\[
\begin{array}{rcccl}
    u^\A[\ov{a}]
 &  =
 &  \big(\Count{\pweight}{\psi_J(\ov{y}})\big)^\A [\ov{a}]
 &  =
 &  \hat{u}_J^\A [\ov{a}]\,.
\end{array}
\]
Hence, the proof of Lemma~\ref{lem:normalform:terms2} is complete.
\end{proof}

\subsubsection{A connected local normalform for \mbox{$\FOW$}}

By combining Lemma~\ref{lem:normalform:terms} with the Gaifman locality
Theorem~\ref{def:GNFforFOW}, we obtain the following normal form for
$\FOW$, which may be of independent interest.
From now on, we assume that whenever $\SC$ contains the integer ring
$(\ZZ,\plus,\mal)$, there is a weight symbol $\weightone\in\Weights$
of arity 1 and type $\ZZ$ such that, in every
$(\sigma,\Weights$)-structure $\A$ that we consider, we have
$\weightone(a)=1\in\ZZ$ for all $a\in A$.

\begin{thm}[cl-Normalform]
\label{cor:normalformFO}\label{thm:clnormalform}
Let $\SC$ contain the integer ring $(\ZZ,\plus,\mal)$.
Every formula $\phi(\ov{x})$ of $\FOWS$
is equivalent to a Boolean
combination of
 $\FOWS$-formulas $\psi(\ov{x})$ that are local
 around their free variables $\ov{x}$,
of local aggregation sentences in $\FOWS$,
and of
 statements of the form ``$g\geq 1$'', for a ground cl-term $g$ of
 type $\ZZ$.

Furthermore, there is an algorithm which transforms an input
$\FOWS$-formula $\phi(\ov{x})$ into an equivalent such formula
$\phi'(\ov{x})$ and outputs the radius of
each ground cl-term in $\phi'$ as well as a number $r$ such that every local
formula in $\phi'$ is $r$-local.
\end{thm}

\begin{proof}
By Theorem~\ref{thm:GaifmanForFOW}, it suffices to translate a basic local
sentence into a statement of the form ``$g\geq 1$'' for a ground
cl-term $g$ of type $\ZZ$.

For a basic local sentence
\ $
 \chi \deff
  \exists y_1\cdots \exists y_k\;
  \vartheta(y_1,\ldots,y_k)
$ \ with \
$  \vartheta(y_1,\ldots,y_k) \deff$
\[
   \Und_{1\leq i<j\leq k}\dist^{\sigma}(y_i,y_j)>2r
   \ \ \und \ \
   \Und_{1\leq i\leq k}\psi(y_i),
\]
let $g_\chi$ be the ground term
\, $
  g_\chi \deff
  \Count{\pweight}{\vartheta(y_1,\ldots,y_k)}$ \
for \,$\pweight\deff \weightone(y_1)\mal\cdots\mal\weightone(y_k)$.

Note that $\vartheta(y_1,\ldots,y_k)$ is $r$-local around its free
variables.
Hence, by Lemma~\ref{lem:normalform:terms}, we obtain a ground cl-term
$\hat{g}_\chi$ such that $\hat{g}_{\chi}^{\A}=g_{\chi}^{\A}$ for every
$(\sigma,\Weights)$-structure $\A$.
Furthermore,
\ $
  \A\models\chi
   \iff
  g_{\chi}^{\A} \geq 1
   \iff
  \hat{g}_{\chi}^{\A}\geq 1\,.
$
This completes the proof of Theorem~\ref{cor:normalformFO}.
\end{proof}

We use the notion \emph{cl-normalform} to denote the
formulas $\phi'(\ov{x})$
provided by Theorem~\ref{cor:normalformFO}.
Note that cl-normalforms do
not necessarily belong to $\FOWS$, but can be viewed as
formulas in $\WAPS$, where $\PP$ contains a unary predicate
$\Pred_{\geq 1}$ of type $\ZZ$ with
$\sem{\Pred_{\geq 1}}\deff \NNpos$.
Then, statements of the form
``$g\geq 1$'' can be expressed via $\Pred_{\geq 1}(g)$.

\subsubsection{A decomposition of $\WAun$-expressions}

Our decomposition of $\WAunPS$ utilises
Theorem~\ref{cor:normalformFO} and is
based on an induction on the maximal nesting depth of term constructions of
the form $\Count{p}{\psi}$ (i.e.\ contructions by rule
\eqref{item:countterm} of Definition~\ref{def:WA}) .
We call this nesting depth the \emph{aggregation depth} (for short:
\emph{$\agg$-depth}) $\countr(\xi)$ of a given formula or term
$\xi$. Formally, $\countr(\phi) $ is defined as follows:
\begin{enumerate}[(1)]
\item
  $\countr(\phi) \ \deff \ 0$, \ if $\phi$ is a formula of the form \
  $x_{1}{=}x_2$ \ or \
  $R(x_1,\ldots,x_{\ar(R)})$
\item
  $\countr((s=\weight(\ov{x}))) \deff \ 0$
\item
  $\countr(\nicht\phi)\ \deff \ \countr(\phi)$
  \ and \
  $\countr((\phi\oder\psi))\ \deff \ \max\set{\countr(\phi),\countr(\psi)}$
\item
  $\countr(\exists y\,\phi) \ \deff \ \countr(\phi)$
\item
  $\countr((s=\sum\weight(\ov{y}).\phi)) \ \deff \ \countr(\phi)$
\item
  $\countr(\Pred(t_1,\ldots,t_m)) \ \deff \ \max\set{\countr(t_1),\ldots,\countr(t_m)}$,
\item
  $\countr(s)\ \deff \ 0$, \ for all $s\in S\in\SC$
\item
  $\countr(\weight(\ov{x})) \ \deff \ 0$
\item
  $\countr((t_1\ast t_2)) \ \deff \
  \max\set{\countr(t_1),\countr(t_2)}$,
   \ for $\ast\in\set{\plus,\minus,\mal}$
\item
  $\countr(\Count{p}{\phi}) \ \deff \ \countr(\phi)+1$.
\end{enumerate}

The base case of our decomposition of $\WAunPS$ is
provided by the following lemma. The proof utilises
Theorem~\ref{cor:normalformFO}.

\begin{lem}\label{lem:normalform:countrOne}
Let $\SC$ contain the integer ring $(\ZZ,\plus,\mal)$.
Let $\phi$ be a $\WAunPS$-formula of the form
$\P(t_1,\ldots,t_m)$ with $\P\in \Ps$,
$m=\ar(\P)$, and where $t_1,\ldots,t_m$ are $\SC$-terms of
$\agg$-depth at most $1$.
Then, $\phi$ is equivalent to a Boolean combination of
\begin{enumerate}[(i)]
\item
 formulas of the form $\P(t'_1,\ldots,t'_{m})$, for cl-terms
 $t'_1,\ldots,t'_{m}$ with $\free(t'_i)=\free(t_i)$ for all
 $i\in[m]$,
\item
 local aggregation sentences in $\FOWS$, and
\item
 statements of the form ``$g\geq 1$'' for ground cl-terms $g$ of type $\ZZ$.
\end{enumerate}
Also, there is an algorithm which transforms an input formula
$\phi$ into such a Boolean combination $\phi'$ and which
outputs the radius of each cl-term and each local formula in $\phi'$.
\end{lem}
\begin{proof}
From Definition~\ref{def:WAun}, we know that either
$\free(\phi)=\emptyset$ or $\free(\phi)=\set{x}$ holds for a variable $x$.
Furthermore, we know that
for every $i\in [m]$, the
$\SC$-term $t_i$ is built by using rules \eqref{item:constterm}--\eqref{item:plustimesterm}
and
$\SC$-terms $\theta'$ of the form~$\Count{\pweight}{\theta}$,
for a $\Weights$-product $\pweight$
such that $\free(\theta)\setminus\varsof(\pweight)\subseteq\set{x}$.
Let $\Theta'$ be the set of all these $\SC$-terms $\theta'$ and let
$\Theta$ be the set of all the according formulas $\theta$.

By assumption, we have $\countr(\phi)\leq 1$. Therefore, every
$\theta\in\Theta$ has $\agg$-depth 0.
Thus, each such $\theta$ is an
$\FOWS$-formula.
By Theorem~\ref{cor:normalformFO}, for each $\theta$ in $\Theta$, we
obtain an equivalent formula $\phi^{(\theta)}$ in cl-normalform.
Let $\Phi$ be the set of all these $\phi^{(\theta)}$.

For each $\theta$ in $\Theta$, the formula
$\phi^{(\theta)}$ is a Boolean
combination of (a) $\FOWS$-formulas that are local around the free variables
of $\theta$, and (b) local aggregation sentences in $\FOWS$, and (c) statements of
the form ``$g\geq 1$'' for a ground cl-term $g$ of type $\ZZ$.

Let $\chi_1,\ldots,\chi_s$ be a
list of all statements of the forms (b) or (c), such that each formula
in $\Phi$ is a Boolean combination of statements in
 $\set{\chi_1,\ldots,\chi_s}$ and of $\FOWS$-formulas that are
 local around their free variables.
For every $J\subseteq [s]$ let \ $\chi_J \;:=\; \Und_{j\in J}\chi_j
\und \Und_{j\in [s]\setminus J}\nicht\chi_j$.

Let $r\in\NN$ be such that each of the local $\FOWS$-formulas that occur in
a formula in $\Phi$ is $r$-local around its free variables.
For each $\theta'$ in $\Theta'$ of the form $\Count{\pweight}{\theta}$, we
apply Lemma~\ref{lem:normalform:terms2} to the term
\[
  t^{(\theta')}
  \ \ \deff \ \
  \Count{\pweight}{\phi^{(\theta)}}
\]
and obtain for every $J\subseteq [s]$ a cl-term
$\hat{t}^{(\theta')}_J$
for which the following is true:

\begin{itemize}
\item
If $\free(\theta')=\emptyset$, then
\ $(\theta')^{\A} = (\hat{t}^{(\theta')}_J)^{\A}$ \
for every $(\sigma,\Weights)$-structure $\A$ with $\A\models\chi_J$.

\item
If $\free(\theta')=\set{x}$, then
\ $(\theta')^{\A}[a] = (\hat{t}^{(\theta')}_J)^{\A}[a]$ \
for every $(\sigma,\Weights)$-structure $\A$ with $\A\models\chi_J$ and every
$a\in A$.
\end{itemize}

\noindent
Thus, for each $J\subseteq [s]$, we have
\[
  \big(\chi_J \ \und \ \P(t_1,\ldots,t_m)\big)
  \ \ \equiv\ \
  \big(\chi_J \ \und \ \P(t_{1,J},\ldots,t_{m,J})\big),
\]
where, for every $i\in [m]$, we let $t_{i,J}$ be the
cl-term obtained from $t_i$ by replacing
each occurrence of a term $\theta'\in\Theta'$ with the
term $\hat{t}^{(\theta')}_J$.
In summary, we obtain the following:
\[
\begin{array}{llllll}
  \phi
& \ \ = \ \
& \P(t_1,\ldots,t_m)\
& \ \  \equiv \ \
& \displaystyle\Oder_{J\subseteq [s]} \big(\
    \chi_J \ \und \ \P(t_1,\ldots,t_m)
  \ \big)
\smallskip\\
&
&
& \ \ \equiv \ \
&  \displaystyle\Oder_{J\subseteq [s]} \big(\
    \chi_J \ \und \ \P(t_{1,J},\ldots,t_{m,J})
  \ \big)
&  =: \ \ \phi'\,.
\end{array}
\]
The formula $\chi_J$ is a Boolean combination of local aggregation
sentences in $\FOWS$ and
 of statements of the
form ``$g\geq 1$'' for ground cl-terms $g$ of type $\ZZ$.
Furthermore, all terms $t_{i,J}$ are cl-terms with
$\free(t_{i,J})\subseteq\free(t_i)$, and we can easily modify them to
achieve that $\free(t_{i,J})=\free(t_i)$.
Thus, the proof of Lemma~\ref{lem:normalform:countrOne} is complete.
\end{proof}

\noindent
We are now ready for the decomposition
theorem for $\WAun$, which can be viewed as a generalisation of the
decomposition theorem for $\FOunC$ provided in \cite{GroheSchweikardt_FOunC}.

\begin{thm}[Decomposition of $\WAun$]\label{thm:normalformWAun}
Let $\SC$ contain the integer ring $(\ZZ,\plus,\mal)$.
Let $z$ be a fixed variable in $\Vars$.
For every $d\in\NN$ and every $\WAunPS$-formula
$\phi(\ov{x})$ of $\agg$-depth $\countr(\phi)=d$, there exists a sequence
$(L_1\ldots,L_{d+1},\phi')$ with the following properties.
\begin{enumerate}[(I)]
\item
  $L_i=(\tau_i,\iota_i)$, for every $i\in \set{1,\ldots,d{+}1}$, where
\begin{enumerate}[$\bullet$]
\item
 $\tau_i$ is a finite set of relation symbols of arity $\leq 1$ that do not belong
 to $\sigma_{i-1}\deff \sigma\cup\bigcup_{j<i}\tau_j$, and
\item
 $\iota_i$ is a mapping that associates with every symbol $R\in\tau_i$
 a formula $\iota_i(R)$
 \begin{enumerate}[(i)]
  \item
   of the form
   $\P(t_1,\ldots,t_{m})$, where $\P\in\Ps$, $m=\ar(\P)$, and
   $t_1,\ldots,t_m$ are cl-terms
   of signature $\sigma_{i-1}$,
   such that $\free(t_j)\subseteq\set{z}$ for each $j\in[m]$, or
  \item
   that is a local aggregation sentence in
   $\FOWSR{\PP}{\sigma_{i-1},\SC,\Weights}$ or a statement
   of the form ``$g\geq 1$'' for a ground cl-term $g$ of
   signature $\sigma_{i-1}$ and of type $\ZZ$.
  \end{enumerate}
  If $R$ has arity 0, then $\iota_i(R)$ has no free variable.
  If $R$ has arity 1, then $z$ is the unique free variable of
   $\iota_i(R)$ (thus, $\iota_i(R)$ is of the form (i)).
\end{enumerate}
\item
$\phi'(\ov{x})$ is a Boolean combination of
     \textup{(A)}~$\FOWSR{\PP}{\sigma_{d+1},\SC,\Weights}$-formulas $\psi(\ov{x})$ that are local around
     their free variables $\ov{x}$, where
     $\sigma_{d+1}\deff\sigma\cup\bigcup_{1\leq i\leq d{+}1}\tau_i$, and
     \textup{(B)}~statements of the form $R()$ where $R$ is a 0-ary relation
     symbol in $\sigma_{d+1}$.
In case that $\free(\phi)=\emptyset$,
$\phi'$ only contains statements of
the latter form.

\item
For every $(\sigma,\Weights)$-interpretation $\I=(\A,\beta)$, we have
$\I\models\phi$ iff $\I_{d+1}\models\phi'$,
where $\I_{d+1}=(\A_{d+1},\beta)$, and $\A_{d+1}$ is the
$\sigma_{d+1}$-expansion of $\A$ defined as follows: \
$\A_0\deff \A$, and for every $i\in[d{+}1]$, $\A_i$ is the
$\sigma_i$-expansion of $\A_{i-1}$, where
for every unary $R\in\tau_i$, we have
\ $
  R^{\A_i} \deff
     \setc{\, a\in A\, }{\, (\A_{i-1},a)\models \iota_i(R)\, }
$ \
and for every 0-ary $R\in\tau_i$ we have
$R^{\A_i}\deff\set{\,\emptytuple\,}$ if $\A_{i-1}\models\iota_i(R)$, and
$R^{\A_i} \deff \emptyset$ if $\A_{i-1}\not\models\iota_i(R)$.
\end{enumerate}
Moreover, there is an algorithm which constructs such a sequence
$D=(L_1,\ldots,L_{d+1},\phi')$ for an input formula $\phi$
and outputs the radius of each cl-term in $D$ as well as a
number $r$ such that every local formula in $\phi'$ is $r$-local
around its free variables.
\end{thm}
\begin{proof}
We proceed by induction on $i$ to construct for all
$i\in [0,d]$
a tuple $L_i=(\tau_i,\iota_i)$ and a $\WAunPSR{\PP}{\sigma_i,\SC,\Weights}$-formula
$\phi_i(\ov{x})$ of $\agg$-depth $(d{-}i)$, such that
for every $(\sigma,\Weights)$-interpretation $\I=(\A,\beta)$ and the
interpretation $\I_i\deff(\A_i,\beta)$,
we have \ $\I\models\phi \iff \I_i\models\phi_i$.

For $i=0$, we are done by letting $\tau_0\deff\emptyset$, $\sigma_0\deff\sigma$, $\phi_0\deff\phi$,
and $\iota_0$ be the mapping with empty domain.
Now assume that for some $i<d$, we have already constructed
$L_{i}=(\tau_{i},\iota_{i})$ and
$\phi_{i}$.
To construct $L_{i+1}=(\tau_{i+1},\iota_{i+1})$ and $\phi_{i+1}$, we proceed as follows.

Let $\Pi$ be the set of all $\WAunPSR{\PP}{\sigma_{i},\SC,\Weights}$-formulas of
$\agg$-depth $\leq 1$ of the form $\P(t_1,\ldots,t_m)$, for $\P\in\Ps$,
that occur in $\phi_{i}$.

Now consider an arbitrary formula $\pi$ in $\Pi$ of the
form $\P(t_1,\ldots,t_m)$.
From Definition~\ref{def:WAun}, we know that there is a variable $y$
such that $\free(t_j)\subseteq\set{y}$ for every $j\in[m]$.
By Lemma~\ref{lem:normalform:countrOne}, $\pi$ is
equivalent to a Boolean combination $\pi'$ of
\begin{enumerate}[(a)]
\item
 formulas of the form $\P(t'_1,\ldots,t'_{m})$, for cl-terms
 $t'_1,\ldots,t'_{m}$ of signature $\sigma_{i}$, where $\free(t'_j)=\free(t_j)\subseteq\set{y}$ for each
 $j\in[m]$,
\item
 statements of the form ``$g\geq 1$'' for ground cl-terms $g$ of
 signature $\sigma_i$, and
\item
 local aggregation sentences in $\FOWSR{\PP}{\sigma_i,\SC,\Weights}$.
\end{enumerate}
For each statement $\chi$ of the form (b) or (c), we include into
$\tau_{i+1}$ a 0-ary relation symbol $R_\chi$,
we replace each occurrence of $\chi$
in $\pi'$ with the new atomic formula $R_\chi()$,
and we let
$\iota_{i+1}(R_\chi)\deff \chi$.
For each statement $\chi$ in $\pi$ of the form (a), we proceed as follows.
If $\free(\chi)=\emptyset$, then we include into
$\tau_{i+1}$ a $0$-ary relation symbol $R_\chi$,
we replace each occurrence of $\chi$
in $\pi'$ with the new atomic formula $R_\chi()$, and
we let $\iota_{i+1}(R_\chi)\deff \chi$.
If $\free(\chi)=\set{y}$, then we include into $\tau_{i+1}$ a unary
relation symbol $R_\chi$,
we replace each
occurrence of $\chi$ in $\pi'$ with the new atomic formula $R_\chi(y)$,
and we let $\iota_{i+1}(R_\chi)$ be the formula
obtained from $\chi$ by consistently replacing every free occurrence
of the variable $y$ with the variable $z$.
We write $\pi''$ for the resulting formula $\pi'$.

Clearly, $\pi''$ is of signature
$\sigma_{i+1}\deff \sigma_i\cup\tau_i$, it has $\agg$-depth~0, and
for every $\sigma$-interpretation $\I=(\A,\beta)$
and $\I_i\deff (\A_i,\beta)$ and $\I_{i+1}\deff (\A_{i+1},\beta)$,
we have:
\ $\I_i\models\pi$ $\iff$ $\I_{i+1}\models\pi''$.

The induction step is completed by letting $\phi_{i+1}$ be
the formula obtained from $\phi_i$ by replacing every occurrence of a
formula $\pi\in\Pi$ with the formula $\pi''$.
It can easily be verified that $\phi_{i+1}$ is an
$\WAunPSR{\Ps}{\sigma_{i+1},\SC,\Weights}$-formula of
$\agg$-depth $\countr(\phi_i)-1 = ((d{-}i){-}1)= (d{-}(i{+}1))$ and that
\ $\I_i\models\phi_i\iff \I_{i+1}\models\phi_{i+1}$.

By the above induction, we have constructed $L_1,\ldots,L_d$ and an
$\WAunPSR{\Ps}{\sigma_d,\SC,\Weights}$-formula $\phi_d$ of $\agg$-depth 0.
Hence, $\phi_d$ is an $\FOWSR{\PP}{\sigma_d,\SC,\Weights}$-formula.
Theorem~\ref{cor:normalformFO} yields an equivalent formula
$\tilde{\phi}$ of signature $\sigma_d$ in cl-normalform.
That is, $\tilde{\phi}$ is a Boolean combination of
\begin{enumerate}[(A)]
 \item $\FOWSR{\PP}{\sigma_d,\SC,\Weights}$-formulas that are local around their free
   variables $\ov{x}$,
 \item local aggregation sentences in
   $\FOWSR{\PP}{\sigma_d,\SC,\Weights}$, and
 \item statements of the form ``$g\geq 1$'', for a ground cl-term $g$
   of type $\ZZ$ and
   of signature $\sigma_d$.
\end{enumerate}
For each statement $\chi$ of the form (B) or (C), we include into
$\tau_{d+1}$ a new relation symbol $R_\chi$ of arity 0,
we replace each occurrence of $\chi$ in $\tilde{\phi}$ with the new
atomic formula $R_\chi()$, and we let
$\iota_{d+1}(R_\chi)\deff \chi$.
Letting $\phi'$ be the resulting formula $\tilde{\phi}$ completes the
proof.
\end{proof}

We call the sequence $(L_1,\ldots,L_{\countr(\xi)+1},\phi')$ that
Theorem~\ref{thm:normalformWAun} provides for a formula
$\phi$ in $\WAunPS$
a \emph{cl-decomposition} of $\phi$.

\subsubsection{Proof of Theorem~\ref{thm:AlgorithmicDecomp}}

By combining Theorem~\ref{thm:normalformWAun} with
Lemmas~\ref{lem:clterm-evaluation} and \ref{lem:normalform:terms},
we can now prove Theorem~\ref{thm:AlgorithmicDecomp}.

\begin{proof}[Proof of Theorem~\ref{thm:AlgorithmicDecomp}] \ \\
Use Theorem~\ref{thm:normalformWAun} to compute a cl-decomposition
$D=(L_1,\ldots,L_{d+1},\phi')$ of $\phi$, for $d\deff \countr(\phi)$.
This formula $\phi'$ is the desired formula.
We let $\sigma_\phi\deff\sigma_{d+1}\deff\sigma\cup\bigcup_{1\leq
  i\leq d{+}1}\tau_i$. We also let $\A^\phi\deff\A_{d+1}$.
To compute $\A^\phi$, we proceed as follows.

Let $\A_0\deff \A$.
For each $i\in[d{+}1]$, compute the $\sigma_i$-expansion of
$\A_{i-1}$.
To achieve this, consider for each $R\in \tau_i$ the formula
$\iota_i(R)$. This formula is of signature $\sigma_{i-1}$ and
(I) of the form $\Pred(t_1,\ldots,t_m)$
for a $\Pred\in\PP$ and cl-terms $t_1,\ldots,t_m$, or
(II) of the form $g\geq 1$ where $g$ is a ground cl-term of type $\ZZ$, or
(III) a local aggregation sentence, i.e.\ of the form $\big( s = \sum
\weight(\ov{y}).\lambda(\ov{y})\big)$ for a local
$\FOWSR{\PP}{\sigma_{i-1},\SC,\Weights}$-formula $\lambda$---and by
Lemma~\ref{lem:normalform:terms},
$\sum\weight(\ov{y}).\lambda(\ov{y})$ is equivalent to a ground cl-term.
Thus, in all three cases,
$\iota_i(R)$ is
a very simple statement that concerns one
or several cl-terms and that involves at most one free variable.
By using Lemma~\ref{lem:clterm-evaluation}, we can
compute in time $|A|\mal d^{\bigOh(1)}$ for each such cl-term $t$ the
values $t^\A[a]$ for all $a\in A$ (resp., the value $t^\A$, if $t$ is
ground).
Then, we combine the values and use a $\Ps$-oracle to check for each
$a\in A$ whether $\iota_i(R)$ is
satisfied by $(\A_{i-1},a)$, and we store the new relation $R^{\A_i}$ accordingly.
\end{proof}

\section{Learning Concepts on Weighted Structures}%
\label{sec: PAC learning}\label{sec:PacLearning}

Throughout this section, fix a
collection \(\SC\) of rings and/or abelian groups,
an $\SC$-predicate collection $(\PP,\ar,\ptype,\sem{{\cdot}})$,
and a finite set \(\Weights\) of weight symbols.

Furthermore, fix numbers $k,\ell\in\NN$.
Let $\LogicL$ be a logic
(e.g.\ $\FO$, $\FOW(\Ps)$, $\WAun(\Ps)$, $\WA(\Ps)$),
let \(\sigma\) be a signature,
and let \(\Phi \subseteq \LogicL[\sigma, \SC, \Weights]\) be a set of formulas
\(\phi(\bar{x}, \bar{y})\) with \(\abs{\bar{x}} = k\) and \(\abs{\bar{y}} = \ell\).
For a \((\sigma, \Weights)\)-structure \(\A\),
we follow the same approach as~\cite{GrienenbergerRitzert_Trees,GroheLoedingRitzert_MSO,GroheRitzert_FO,GroheTuran_Learnability,vanBergerem}
and consider the instance space \(\X = A^k\)
and concepts from the concept class
\[
\C(\Phi, \A, k, \ell) \ \coloneqq \ \bigsetc{ \; \llbracket
 \phi(\bar{x}, \bar{y}) \rrbracket^\A (\bar{x}, \bar{v})\;}{\; \phi
 \in \Phi,\ \bar{v} \in A^\ell \;},
\]
where
$\sem{\phi(\ov{x},\ov{y})}^\A(\ov{x},\ov{v})$ is defined as the mapping from $A^k$
to $\set{0,1}$ that maps $\ov{a}\in A^k$ to $\sem{\phi(\ov{a},\ov{v})}^\A$,
which is $1$ iff $\A\models\phi[\ov{a},\ov{v}]$.
Given a training sequence
\(T = \big((\ov{a}_1, b_1), \dots, (\ov{a}_t, b_t)\big)\)
from \((A^k \times \{0,1\})^t\),
we want to compute a hypothesis
that consists of a formula \(\phi\) and a tuple of parameters \(\bar{v}\)
and is, depending on the approach,
consistent with the training sequence or
probably approximately correct.

Instead of allowing random access to the background structure,
we limit our algorithms to have only \emph{local access}.
That is, an algorithm may only interact with the structure via queries of the form
``Is \(\bar{a} \in R^\A\)?'',
``Return $\weight^\A(\ov{a})$''
and ``Return a list of all neighbours of $a$ in the Gaifman graph of $\A$''.
Hence, in this model, algorithms are required to access new vertices
only via neighbourhood queries of vertices they have already seen.
This enables us to learn a concept from examples
even if the background structure is too large to fit into the main memory.
To obtain a reasonable running time,
we intend to find algorithms that compute a hypothesis
in sublinear time, measured in the size of the background structure.
This local access model has already been studied for relational structures
in~\cite{GroheRitzert_FO,vanBergerem}
for concepts definable in $\FO$ or in \(\FOCNP\).
Modifications of the local access model for strings
and trees have been studied
in~\cite{GrienenbergerRitzert_Trees,GroheLoedingRitzert_MSO}.

In many applications, the same background structure is used
multiple times to learn different concepts.
Hence, similar to the approaches in~\cite{GrienenbergerRitzert_Trees,GroheLoedingRitzert_MSO},
we allow a precomputation step to enrich the background structure with additional information.
That is, instead of learning on a \((\sigma, \Weights)\)-structure \(\A\),
we use an enriched \((\sigma^*, \Weights)\)-structure \(\A^*\),
which has the same universe as $\A$, but
\(\sigma^* \supseteq \sigma\) contains additional relation symbols.
The hypotheses we compute may make use of this additional information and thus,
instead of representing them via formulas from the fixed set \(\Phi\),
we consider a set \(\Phi^*\) of formulas of signature
$\sigma^*$.
These formulas may even belong to a logic $\LogicL^*$ different from
$\LogicL$.
We study the following learning problem.

\begin{problem}[Exact Learning with Precomputation]%
  \label{problem:exact_learning}
  Let \(\Phi \subseteq \LogicL[\sigma, \SC, \Weights]\) and
  \(\Phi^* \subseteq \LogicL^*[\sigma^*, \SC, \Weights]\)
  such that, for every \((\sigma, \Weights)\)-structure \(\A\),
  there is a \((\sigma^*, \Weights)\)-structure \(\A^*\)
  with \(U(\A^*) = U(\A)\)
  that satisfies \(\C(\Phi, \A, k, \ell) \subseteq \C(\Phi^*, \A^*, k, \ell)\),
  i.e.\ every concept that can be defined on \(\A\) using \(\Phi\)
  can also be defined on \(\A^*\) using \(\Phi^*\).
  The task is as follows.
  \begin{description}
    \item[Given] a training sequence
      \(T = \big((\bar{a}_1, b_1), \dots, (\bar{a}_t, b_t)\big) \in (A^k \times \{0,1\})^t\)
      and, for a \((\sigma, \Weights)\)-structure \(\A\),
      local access to the associated \((\sigma^*, \Weights)\)-structure \(\A^*\),
    \item[return] a formula \(\phi^* \in \Phi^*\)
      and a tuple \(\bar{v} \in A^\ell\) of parameters
      such that the hypothesis \(\llbracket \phi^*(\bar{x}, \bar{y}) \rrbracket^{\A^*} (\bar{x}, \bar{v})\)
      is consistent with $T$,
      i.e.\ it maps $\ov{a}_i$ to $b_i$ for every $i\in[t]$.

      The algorithm may reject if there is no consistent classifier using a formula from \(\Phi\) on \(\A\).
    \end{description}
\end{problem}

Next, we examine requirements for \(\Phi\) and \(\Phi^*\)
that help us solve \cref{problem:exact_learning} efficiently.
Following the approach presented in~\cite{GroheRitzert_FO},
to obtain algorithms that run in sublinear time,
we study concepts that can be represented via
a set of \emph{local} formulas \(\Phi\)
with a finite set \(\Phi^*\) of normal forms.
Using Feferman-Vaught decompositions
and the locality of the formulas,
we can then limit the search space for the parameters
to those that are in a certain neighbourhood
of the training sequence.
Recall that \(\Phi\) is a set of formulas
\(\phi(\bar{x}, \bar{y})\) in \(\LogicL[\sigma, \SC, \Weights]\)
with \(\abs{\ov{x}} = k\) and \(\abs{\ov{y}}=\ell\).
In the following, we require \(\Phi\) to have the following property.

\begin{prop}%
  \label{prop: requirements PAC learning}
  There are a signature \(\sigma^*\),
  a logic \(\LogicL^*\),
  an \(r \in \NN\), and a finite set of \(r\)-local formulas
  \(\Phi^* \subseteq \LogicL^*[\sigma^*, \SC, \Weights]\) such that
  the following hold.
  \begin{enumerate}[(1)]
    \item\label{prop: requirements PAC learning: associated structure}
      For every \((\sigma, \Weights)\)-structure \(\A\), there is
      a \((\sigma^*, \Weights)\)-structure \(\A^*\) with \(U(\A^*) = U(\A)\)
      such that, for every \(\phi(\bar{x}, \bar{y}) \in \Phi\),
      there is a \(\phi^*(\bar{x}, \bar{y}) \in \Phi^*\)
      with \(\A \models \phi[\bar{a}, \bar{b}] \iff \A^* \models \phi^*[\bar{a}, \bar{b}]\)
      for all \(\bar{a} \in A^k\), \(\bar{b} \in A^\ell\).
    \item\label{prop: requirements PAC learning: FV}
      Every \(\phi^* \in \Phi^*\) has, for every partition
      $(\ov{z};\ov{z}')$ of the free variables of $\phi^*$,
      a Feferman-Vaught decomposition in \(\Phi^*\) w.r.t.\ $(\ov{z};\ov{z}')$.
    \item\label{prop: requirements PAC learning: Boolean combinations}
      For all $\phi^*_1,\phi^*_2\in\Phi^*$, the set $\Phi^*$ contains formulas
      equivalent to $\nicht\phi^*_1$ and to $(\phi^*_1\oder\phi^*_2)$.
  \end{enumerate}
\end{prop}

This property suffices to solve \cref{problem:exact_learning}:

\begin{thm}[Exact Learning with Precomputation]%
  \label{thm: exact learning}
  There is an algorithm that solves
  \cref{problem:exact_learning}
  with local access to a structure \(\A^*\)
  associated with a structure \(\A\)
  in time \(f_{\Phi^*}(\A^*) \cdot \big(\log n + d + t\big)^{\bigO(1)}\),
  where \(\A\), \(\A^*\), \(\Phi\), and \(\Phi^*\) are as described in \cref{prop: requirements PAC learning},
  \(t\) is the number of training examples,
  \(n\) and \(d\) are the size and the degree
  of \(\A^*\),
  and \(f_{\Phi^*}(\A^*)\) is an upper bound on the
  time complexity of model checking
  for formulas in \(\Phi^*\) on \(\A^*\).
\end{thm}

We prove the theorem in \cref{sec: Exact Learning with Precomputation}.

Apart from exact learning with precomputation,
we also study hypotheses that
generalise well in the following sense.
The \emph{generalisation error} of a hypothesis
\(h \colon A^k \to \{0,1\}\) for a probability distribution
\(\D\) on \(A^k \times \{0,1\}\) is
\[\mathrm{err}_\D (h) \ \coloneqq \ \underset{(\bar{a}, b) \sim \D}{\Pr}(h(\bar{a}) \neq b).\]

We write $\openinterval{0}{1}$ for the set of all rationals
$q$ with $0<q<1$.
A hypothesis class \(\mathcal{H} \subseteq \{0,1\}^{A^k}\)
is \emph{agnostically PAC-learnable} if there is a function
\(t_\mathcal{H} \colon \openinterval{0}{1}^2 \to \NN\) and
a learning algorithm \(\mathfrak{L}\) such that for all
\(\epsilon, \delta \in \openinterval{0}{1}\) and for every distribution \(\D\)
over \(A^k \times \{0,1\}\), when running \(\mathfrak{L}\)
on a sequence \(T\) of \(t_\mathcal{H} (\epsilon, \delta)\) examples
drawn i.i.d.\ from \(\D\),
it holds that
\begin{equation*}
  \Pr \left( \mathrm{err}_\D (\mathfrak{L}(T)) \leq \inf_{h \in \mathcal{H}} \mathrm{err}_\D (h) + \epsilon \right) \geq 1 - \delta.
\end{equation*}

The following theorem, which we prove in \cref{sec: agnostic PAC
  learning with precomputation}, provides an agnostic PAC learning algorithm.

\begin{thm}[Agnostic PAC Learning with Precomputation]%
  \label{thm: agnostic PAC learning with precomputation}
  Let \(\A\), \(\A^*\), and \(\Phi^*\) be
  as in \cref{prop: requirements PAC learning}.
  There is an \(s \in \NN\) such that,
  given local access to \(\A^*\),
  the hypothesis class \(\mathcal{H} \coloneqq \C(\Phi^*, \A^*, k, \ell)\) is agnostically PAC-learnable
  with \(t_\mathcal{H} (\epsilon, \delta) = s \cdot \left\lceil \frac{\log(n/\delta)}{\epsilon^2} \right\rceil\)
  via an algorithm that,
  given \(t_\mathcal{H} (\epsilon, \delta)\) examples,
  returns a hypothesis of the form \((\phi^*, \bar{v}^*)\)
  with \(\phi^* \in \Phi^*\) and \(\bar{v}^* \in A^\ell\)
  in time \(f_{\Phi^*}(\A^*) \cdot \big(\log n + d + \frac{1}{\epsilon} + \log \frac{1}{\delta}\big)^{\bigO(1)}\) with
  only local access to \(\A^*\),
  where \(n\) and \(d\) are the size and the degree of \(\A^*\),
  and \(f_{\Phi^*}(\A^*)\) is an upper bound on the
  time complexity of model checking
  for formulas in \(\Phi^*\) on \(\A^*\).
\end{thm}

The next remark establishes the crucial link between the learning results
of this section and the locality results of
Section~\ref{sec:Locality}: it shows that suitably chosen sets
$\Phi\subseteq\WAunPS$ indeed have Property~\ref{prop: requirements PAC learning}.
\begin{rem}%
  \label{remark:propertyForLearning}
  Fix a $q\in\NN$ and let $\Phi\deff\Phi_{q,k+\ell}$ be the set of all
  $\FO[\sigma]$-formulas $\phi$ of quantifier rank at most $q$ and
  with free variables among $\set{x_1,\ldots,x_k,y_1,\ldots,y_\ell}$. By
  the well-known properties of first-order logic,
  $\Phi$ has \cref{prop: requirements PAC learning}
  (e.g.\ via $L' \coloneqq L =\FO$, $\sigma^* \coloneqq \sigma$, and $\A^* \coloneqq \A$; this is exactly the
  setting considered in~\cite{GroheRitzert_FO}).
  By using the locality properties of $\FOW$ and
  $\WAun$ from \cref{subsec:Locality}, we can apply a
  similar reasoning to $\WAunPS$ as to $\FO[\sigma]$:
  let the collections $\Ps$ and $\SC$ be finite (but $\SC$ may contain
  some infinite rings or abelian groups), fix
  a finite set $\mathscr{S}$ of elements $s\in S\in\SC$, and fix
  a $q\in\NN$.
  Let $\Phi\deff\Phi_{q,k+\ell,\mathscr{S}}$ be the set of all
  $\WAunPS$-formulas $\phi$ of quantifier rank and aggregation depth at
  most $q$ and with free variables among
  $\set{x_1,\ldots,x_k,y_1,\ldots,y_\ell}$ that have the following
  additional property: all symbols $s\in S\in\SC$ that are present in
  $\phi$ belong to $\mathscr{S}$, all $\Weights$-products present in
  $\phi$ have length at most $q$, and the maximum nesting depth of term
  constructions using rule~\eqref{item:plustimesterm} in order to
  construct terms present in $\phi$ is at most $q$.
  We claim that this set $\Phi$ has \cref{prop: requirements PAC learning}.
  To prove this, we can argue as follows.

  \begin{claim}\label{claim:finitePhi}
    Up to logical equivalence, $\Phi$ only contains a finite number
    of formulas.
  \end{claim}

  \begin{claimproof}
    Since the maximum nesting depth of constructs using
    rules~\eqref{item:exists} and~\eqref{item:finitegroup}
    as well as the maximum nesting depth of constructs using
    rule~\eqref{item:countterm} from \cref{def:WA}
    is bounded by \(q\)
    and every construct using
    rule~\eqref{item:exists} adds one new variable,
    rule~\eqref{item:finitegroup} adds at most \(q\) new variables, and
    rule~\eqref{item:countterm} adds at most \(q^2\) new variables,
    every subformula of a formula in \(\Phi\)
    has at most \(k + \ell + q^2 + q^3\) free variables.
    With finitely many free variables and
    \(\sigma\), \(\mathscr{S}\), and \(\Weights\) being finite as well,
    rules~\eqref{item:atomic} and~\eqref{item:wsimple} only produce a finite number of
    formulas.
    With the same argument, rules~\eqref{item:constterm}
    and~\eqref{item:wsimpleterm} only produce a finite number of \(\SC\)-terms.
    We show by induction on the nesting depth of constructs
    using rules~\eqref{item:exists}, \eqref{item:finitegroup},
    and~\eqref{item:countterm}
    that there are, up to logical equivalence,
    only finitely many (sub-)formulas and \(\SC\)-terms used in \(\Phi\),
    which implies the claim.

    If there are only finitely many \(\SC\)-terms,
    then, with a bounded nesting depth,
    rule~\eqref{item:plustimesterm} only yields
    a finite number of new \(\SC\)-terms.
    Thus, since \(\Ps\) is also finite,
    rule~\eqref{item:Q} only produces a finite number of formulas of the form
    \(\Pred(t_1,\ldots,t_m)\).
    Hence, rule~\eqref{item:bool} only creates
    a finite number of formulas up to logical equivalence.
    (Consider them being in a normal form analogous to CNF.)

    Applying rule~\eqref{item:exists}
    or rule~\eqref{item:finitegroup}
    to a set of finitely many formulas
    only creates finitely many new formulas.
    Then, rule~\eqref{item:countterm}
    only yields finitely many \(\SC\)-terms.
    This completes the proof of Claim~\ref{claim:finitePhi}.
  \end{claimproof}

  For each of these finitely many formulas $\phi$, we apply
  \cref{thm:AlgorithmicDecomp} to obtain an extension
  $\sigma_\phi$ of $\sigma$, a $\sigma_\phi$-expansion $\A^\phi$ of
  $\A$, and a local $\FOWSR{\Ps}{\sigma,\SC,\Weights}$-formula $\phi'$.
  Then we let $\sigma^*$ be the union of all the $\sigma_\phi$, we
  let $\A^*$ be the $\sigma^*$-expansion of $\A$ whose
  $\sigma_\phi$-reduct coincides with $\A^\phi$ for every $\phi$,
  and we let $\Phi'$ be the set of all the formulas $\phi'$.
  Choose a number $r\in\NN$ such that each of the $\phi'\in\Phi'$ is
  $r$-local.

  We can repeatedly apply \cref{thm:FVforFOW},
  take the \(r\)-localisations \(\alpha^{(r)},\beta^{(r)}\)
  of the resulting formulas \(\alpha,\beta\),
  and take Boolean combinations
  to obtain an extension
  \(\Phi^*\) of \(\Phi'\) such that \(\Phi^*\) satisfies
  statements~\eqref{prop: requirements PAC learning: FV}
  and~\eqref{prop: requirements PAC learning: Boolean combinations} of
  \cref{prop: requirements PAC learning} and contains only
  \(r\)-local formulas.

  \begin{claim}\label{claim:stopProcess}
    One can stop the process after finitely many steps
    and thus obtain a finite extension \(\Phi^*\).
  \end{claim}

  \begin{claimproof}
    When applying \cref{thm:FVforFOW} to a formula \(\phi\)
    w.r.t.\ \((\ov{x};\ov{y})\),
    the Feferman-Vaught decomposition only contains
    new formulas if
    \(\{x_1,\ldots,x_k\} \subsetneq \free(\phi)\) and
    \(\{y_1,\ldots,y_\ell\} \subsetneq \free(\phi)\).
    Hence, if one only applies \cref{thm:FVforFOW}
    and takes the \(r\)-localisations of the resulting formulas,
    then one can stop the process after finitely many steps.
    Let \(\Phi^{(0)} = \Phi' , \Phi^{(1)}, \ldots, \Phi^{(m)}\)
    be the resulting sets of formulas from this finite process.

    Let \(\phi\) be a Boolean combination of formulas from \(\Phi^{(i)}\).
    If \(\phi = \phi_1 \lor \phi_2\)
    with \(\phi_1, \phi_2 \in \Phi^{(i)}\), then
    \(\Delta_\phi = \Delta_{\phi_1} \cup \Delta_{\phi_2}
      \subseteq \left(\Phi^{(i+1)}\right)^2\).
    If \(\phi = \nicht \psi\) with \(\psi \in \Phi^{(i)}\)
    and \(\Delta_\psi = \{(\alpha_1, \beta_1), \dots, (\alpha_s, \beta_s)\}\),
    then
    \(\Delta_\phi = \setc{(\alpha_A, \beta_{[s] \setminus A})}{A \subseteq [s]}
      \subseteq \left(\Phi^{(i+1)}\right)^2\)
    with \(\alpha_A = \bigwedge_{i \in A} \neg \alpha_i\)
    and \(\beta_A = \bigwedge_{i \in A} \neg \beta_i\).
    Inductively, it follows that all formulas
    used in the Feferman-Vaught decomposition
    of \(\phi\) w.r.t.\ \((\ov{x};\ov{y})\) are Boolean combinations
    of formulas in \(\Phi^{(i+1)}\).

    Hence, the result of the overall process
    is the set of Boolean combinations
    of \(r\)-localisations of formulas in \(\Phi^{(m)}\).
    Since the set of Boolean combinations of finitely many formulas is,
    up to logical equivalence, again finite,
    the process stops after finitely many steps
    with a finite extension \(\Phi^*\).
    This completes the proof of Claim~\ref{claim:stopProcess}.
  \end{claimproof}

\noindent
  This $\Phi^*$ witnesses that $\Phi$ has
  \cref{prop: requirements PAC learning}.
\end{rem}

\subsection{Exact Learning with Precomputation}%
\label{sec: Exact Learning with Precomputation}

Section~\ref{sec: Exact Learning with Precomputation}
is devoted to the proof of Theorem~\ref{thm: exact learning}.

Let \(\A\) be a \((\sigma, \Weights)\)-structure and
let \(\A^*\) and \(\Phi^*\) be as in \cref{prop: requirements PAC learning}.
To prove \cref{thm: exact learning}, we present an algorithm that
follows similar ideas as the algorithm presented
in~\cite{GroheRitzert_FO}. Note, however, that~\cite{GroheRitzert_FO}
focuses on first-order logic, whereas our setting allows to achieve results for
considerably stronger logics.

While the set of possible formulas \(\Phi^*\) already has constant size,
we have to reduce the parameter space
to obtain an algorithm that runs in sublinear time.
Since the formulas in \(\Phi^*\) are \(r\)-local,
we show that it suffices to consider parameters
in a neighbourhood of the training sequence with a fixed radius.
The main ingredient
is the following result, which uses a
Feferman-Vaught decomposition
and allows us to analyse the parameters we choose by splitting them
into two parts with disjoint neighbourhoods.
For \(\bar{a} \in A^{\abs{\bar{z}}}\), let
\(\type{\Phi^*}{\A^*}{\bar{a}} \coloneqq \{\phi^*(\bar{z}) \in \Phi^* : \A^* \models \phi^*[\bar{a}]\}\).

\begin{lem}[Local Composition Lemma]%
  \label{lem: local composition lemma}
  For numbers $k',\ell'$,
  let \(\bar{a}, \bar{a}' \in A^{k'}\),
  \(\bar{b}, \bar{b}' \in A^{\ell'}\),
  \(\dist^{\A^*}(\bar{a}, \bar{a}') > 2r{+}1\),
  \(\dist^{\A^*}(\bar{b}, \bar{b}') > 2r{+}1\),
  \(\type{\Phi^*}{\A^*}{\bar{a}} = \type{\Phi^*}{\A^*}{\bar{a}'}\), and
  \(\type{\Phi^*}{\A^*}{\bar{b}} = \type{\Phi^*}{\A^*}{\bar{b}'}\).
  Then
  \(\type{\Phi^*}{\A^*}{\bar{a}\bar{b}} = \type{\Phi^*}{\A^*}{\bar{a}'\bar{b}'}\).
\end{lem}

\begin{proof}
  Let \(\phi^*(\bar{x},\bar{y}) \in \type{\Phi^*}{\A^*}{\bar{a}\bar{b}}\).
  Then, with \cref{prop: requirements PAC learning}\,(\ref{prop: requirements PAC learning: FV}),
  \(\phi^*\) has a Feferman-Vaught decomposition \(\Delta\) in
  \(\Phi^*\) w.r.t.\ $(\ov{x};\ov{y})$, and thus,
  \(\Neighb{r}{\A^*}{\bar{a}} \oplus \Neighb{r}{\A^*}{\bar{b}} \models \phi^*[\bar{a}, \bar{b}]\) if and only if
  there exists \((\alpha, \beta) \in \Delta\) such that
  \(\Neighb{r}{\A^*}{\bar{a}} \models \alpha[\bar{a}]\) and \(\Neighb{r}{\A^*}{\bar{b}} \models \beta[\bar{b}]\).
  Since \(\A^* \models \phi^*[\bar{a}, \bar{b}]\) and
  $\phi^*,\alpha,\beta$
  are \(r\)-local,
  it follows that \(\A^* \models \alpha[\bar{a}]\) and \(\A^* \models \beta[\bar{b}]\).
  Hence, \(\alpha \in \type{\Phi^*}{\A^*}{\bar{a}} = \type{\Phi^*}{\A^*}{\bar{a}'}\) and
  \(\beta \in \type{\Phi^*}{\A^*}{\bar{b}} = \type{\Phi^*}{\A^*}{\bar{b}'}\).
  We obtain
  \(\A^* \models \bigvee_{(\alpha, \beta) \in \Delta} \alpha[\bar{a}'] \land \beta[\bar{b}']\)
  and thus \(\A^* \models \phi^*[\bar{a}', \bar{b}']\).
\end{proof}

The next lemma shows that it suffices to
search in a reduced parameter space to find
a consistent hypothesis.
For $S\subseteq A$ and an element $b\in A$,
let $\dist^{\A^*}(b,S)\deff\min_{a\in S}\dist^{\A^*}(b,a)$.
For
$R\geq 0$, set $\neighb{R}{\A^*}{S}\deff\bigcup_{a\in
  S}\neighb{R}{\A^*}{a}$. Also, for a training sequence
\(T = \big((\bar{a}_1, b_1), \dots, (\bar{a}_t, b_t)\big) \in (A^k
\times \{0,1\})^t\), let
$\neighb{R}{\A^*}{T}\deff\neighb{R}{\A^*}{S}$, where $S$ is the set of
all $a\in A$ that occur in one of the $\ov{a}_i$.

\begin{lem}%
  \label{lem: existence of consistent local concept}
  Let \(T = \big((\bar{a}_1, b_1), \dots, (\bar{a}_t, b_t)\big) \in (A^k \times \{0,1\})^t\) be
  consistent with some classifier in \(\C(\Phi^*, \A^*, k, \ell)\).
  Then there are a formula \(\phi^*(\bar{x}, \bar{y}) \in \Phi^*\) and a tuple
  \(\bar{v}^* \in \neighb{(2r+1)\ell}{\A^*}{T}^\ell\) such that
  \(\llbracket\phi^*(\bar{x}, \bar{y})\rrbracket^{\A^*}(\bar{x}, \bar{v}^*)\) is consistent with \(T\).
\end{lem}

\begin{proof}
  The proof is similar to the proof of the analogous statement
  in~\cite{GroheRitzert_FO} for the special case of $\FO$, but relies on
  \cref{prop: requirements PAC learning} and
  \cref{lem: local composition lemma}.

  Let \(\phi(\bar{x}, \bar{y}) \in \Phi^*\)
  and \(\bar{v}=(v_1, \dots, v_\ell) \in A^\ell\) such that
  \(\llbracket\phi(\bar{x}, \bar{y})\rrbracket^{\A^*}(\ov{x},\ov{v}) \in \C(\Phi^*, \A^*, k, \ell)\)
  is consistent with \(T\).
  Let \(N^{(0)} \coloneqq \neighb{r}{\A^*}{T}\).
  Now we inductively define \(v^{(i)}\) and \(N^{(i)}\) for \(i \geq 1\) as follows.
  Given \(N^{(i-1)}\), if there is a
  \(v \in \{v_1, \dots, v_\ell\} \setminus \{v^{(1)}, \dots, v^{(i-1)}\}\)
  such that \(\dist^{\A^*}(v, N^{(i)}) \leq r{+}1\),
  then we set \(v^{(i)} \coloneqq v\) and
  \(N^{(i)} \coloneqq N^{(i-1)} \cup \neighb{r}{\A^*}{v^{(i)}}\). If there is
  no such \(v\), then we set \(m \coloneqq i{-}1\) and stop.

  W.l.o.g.\
  let \(v^{(i)} = v_i\) for \(i \in [m]\).
  Let \(\bar{v}^< \coloneqq (v_1, \dots, v_m)\)
  and \(\bar{v}^> \coloneqq (v_{m+1}, \dots, v_\ell)\).
  Then \(\bar{v}^< \in \big(\neighb{(2r+1)\ell}{\A^*}{T}\big)^m\).

  \begin{claim}
    \label{claim: same local type then same eval}
    Let \(i,j \in [t]\) such that
    \(\type{\Phi^*}{\A^*}{\bar{a}_i \bar{v}^<} = \type{\Phi^*}{\A^*}{\bar{a}_j \bar{v}^<}\).
    Then \(b_i = b_j\).
  \end{claim}

  \begin{claimproof}
    From the construction, it follows that \(\dist^{\A^*}(\bar{v}^>, N^{(m)}) > r{+}1\)
    and \(\neighb{r}{\A^*}{\bar{v}^<} \subseteq N^{(m)}\).
    Hence, with \(\neighb{r}{\A^*}{\bar{a}_s} \subseteq N^{(m)}\),
    we obtain \(\dist^{\A^*}(\bar{v}^>, \bar{a}_s \bar{v}^<) > 2r{+}1\)
    for every \(s \in [t]\).
    With \cref{lem: local composition lemma}, it follows that
\ $
\type{\Phi^*}{\A^*}{\bar{a}_i \bar{v}} = \type{\Phi^*}{\A^*}{\bar{a}_i \bar{v}^< \bar{v}^>} = \type{\Phi^*}{\A^*}{\bar{a}_j \bar{v}^< \bar{v}^>} = \type{\Phi^*}{\A^*}{\bar{a}_j \bar{v}}$.
    Thus, in particular,
    \(\phi \in \type{\Phi^*}{\A^*}{\bar{a}_i \bar{v}}\)
    $\iff$\(\phi \in \type{\Phi^*}{\A^*}{\bar{a}_j \bar{v}}\). Since
    $\llbracket\phi(\bar{x}, \bar{y})\rrbracket^{\A^*}(\ov{x},\ov{v})$
    is consistent with $T$, this implies that $b_i=b_j$.
  \end{claimproof}
  We let \(\bar{y}^< \coloneqq (y_1, \dots, y_m)\) and choose
  \[\phi^< \ \coloneqq \bigvee_{i \in [t], b_i=1} \quad \bigwedge_{\gamma(\bar{x},\bar{y}^<) \,\in\, \type{\Phi^*}{\A^*}{\bar{a}_i\bar{v}^<}} \gamma(\bar{x},\bar{y}^<)\]
  The formula $\phi^<$ is a
  Boolean combination of formulas in \(\Phi^*\) and thus,
  according to \cref{prop: requirements PAC learning},
  there is a formula \(\phi^* \in \Phi^*\)
  that is equivalent to \(\phi^<\).
  The free variables of $\phi^*$ are among $\ov{x}$ and $\ov{y}^<$, and
  since $\ov{y}^<$ is a prefix of $\ov{y}$, we can safely write $\phi^*(\ov{x},\ov{y})$.
  We turn $\ov{v}^<=(v_1,\ldots,v_m)$ into a tuple $\ov{v}^*\in
  \neighb{(2r+1)\ell}{\A^*}{T}^\ell$ by
  choosing an arbitrary \(v \in \neighb{(2r+1)\ell}{\A^*}{T}\) and
  filling the missing $(\ell{-}m)$ positions with the value $v$.

  By the choice of $\phi^<$, the following is true for all \(j \in [t]\):
  if
  \(\A^* \models \phi^*[\bar{a}_j, \bar{v}^*]\),
  then there is a positive example \(\bar{a}_i\) with
  \(\A^* \models \bigwedge_{\gamma(\bar{x},\bar{y}^<) \in \type{\Phi^*}{\A^*}{\bar{a}_i\bar{v}^<}} \gamma[\bar{a}_j,\bar{v}^<]\).
  Thus
  \(\type{\Phi^*}{\A^*}{\bar{a}_i\bar{v}^<} = \type{\Phi^*}{\A^*}{\bar{a}_j\bar{v}^<}\)
  for some positive example \(\bar{a}_i\);
  and with \cref{claim: same local type then same eval}, we can conclude that \(b_j=1\).
  Conversely, if \(b_j=1\),
  then \(\A^* \models \bigwedge_{\gamma(\bar{x},\bar{y}^<) \in \type{\Phi^*}{\A^*}{\bar{a}_j\bar{v}^<}} \gamma[\bar{a}_j,\bar{v}^<]\)
  and hence \(\A^* \models \phi^*[\bar{a}_j, \bar{v}^*]\).
  Thus, \(\llbracket\phi^*(\bar{x}, \bar{y})\rrbracket^{\A^*}(\bar{x}, \bar{v}^*)\) is consistent with \(T\).
\end{proof}

We can now prove \cref{thm: exact learning}.

\begin{figure}
\begin{minipage}{7cm}
  \begin{algorithmic}[1]
    \State{\(N \leftarrow \neighb{(2r+1)\ell}{\A^*}{T}\)}
    \ForAll{\(\bar{v}^* \in N^\ell\)}
      \ForAll{\(\phi^*(\bar{x}, \bar{y}) \in \Phi^*\)}
        \State{\(consistent \leftarrow \True\)}
        \ForAll{\(i \in [t]\)}
          \State{\(\Structure{N} = \Neighb{r}{\A^*}{\bar{a}_i \bar{v}^*}\)}
          \If{\(\llbracket\phi^*(\bar{a}_i, \bar{v}^*)\rrbracket^{\Structure{N}} \neq b_i\)}
            \State{\(consistent \leftarrow \False\)}
          \EndIf
        \EndFor
        \If{\(consistent\)}
          \Return{\((\phi^*,\,\bar{v}^*)\)}
        \EndIf
      \EndFor
    \EndFor
    \Reject
  \end{algorithmic}
\end{minipage}
\hspace{1.5em}
\begin{minipage}{8cm}
  \begin{algorithmic}[1]
    \State{\(N \leftarrow \neighb{(2r+1)\ell}{\A^*}{T}\)}
    \State{\(err_{\min} \leftarrow |T| + 1\)}
    \ForAll{\(\bar{v}^* \in N^\ell\)}
      \ForAll{\(\phi^*(\bar{x}, \bar{y}) \in \Phi^*\)}
        \State{\(err_{\textup{cur}} \leftarrow 0\)}
        \ForAll{\(i \in [t]\)}
          \State{\(\Structure{N} = \Neighb{r}{\A^*}{\bar{a}_i \bar{v}^*}\)}
          \If{\(\llbracket\phi^*(\bar{a}_i, \bar{v}^*)\rrbracket^{\Structure{N}} \neq b_i\)}
            \State{\(err_{\textup{cur}} \leftarrow err_{\textup{cur}} + 1\)}
          \EndIf
        \EndFor
        \If{\(err_{\textup{cur}} < err_{\min}\)}
          \State{\(err_{\min} \leftarrow err_{\textup{cur}}\)}
          \State{\(\phi^*_{\min} \leftarrow \phi^*\)}
          \State{\(\bar{v}^*_{\min} \leftarrow \bar{v}^*\)}
        \EndIf
      \EndFor
    \EndFor
    \Return{\((\phi_{\min}^*,\,\bar{v}_{\min}^*)\)}
  \end{algorithmic}
\end{minipage}
\caption{Learning algorithms for Theorems~\ref{thm: exact learning} (left)
  and~\ref{thm: agnostic PAC learning with precomputation} (right).
  Both algorithms use as input a
  training sequence $T = \big((\bar{a}_1, b_1), \dots, (\bar{a}_t,
  b_t)\big) \in (A^k \times \{0,1\})^t$ and have only local access to the
  structure $\A^*$.
}%
\label{fig: exact learning}
\label{fig: agnostic PAC learning with precomputation}
\end{figure}

\begin{proof}[Proof of \cref{thm: exact learning}]
  We show that the algorithm depicted on the left-hand side of  \cref{fig: exact learning} fulfils the requirements given in \cref{thm: exact learning}.
  The algorithm goes through all tuples
  \(\bar{v}^* \in (\neighb{(2r+1)\ell}{\A^*}{T})^\ell\)
  and all formulas \(\phi^*(\bar{x}, \bar{y}) \in \Phi^*\).
  A hypothesis  \(\llbracket\phi^*(\bar{x}, \bar{y})\rrbracket^{\A^*}(\bar{x}, \bar{v}^*)\)
  is consistent with the training sequence \(T\) if and only if
  \(\llbracket\phi^*(\bar{a}_i, \bar{v}^*)\rrbracket^{\A^*} = b_i\)
  for all \(i \in [t]\).
  Since \(\Phi^*\) only contains \(r\)-local formulas,
  this holds if and only if
  \(\llbracket\phi^*(\bar{a}_i, \bar{v}^*)\rrbracket^{\Neighb{r}{\A^*}{\bar{a}_i \bar{v}^*}} = b_i\)
  for every \(i \in [t]\).
  Hence, the algorithm only returns a hypothesis if it is consistent.
  Furthermore, if there is a consistent hypothesis in \(\C(\Phi, \A, k, \ell)\),
  then by \cref{prop: requirements PAC learning}\,(\ref{prop: requirements PAC learning: associated structure}),
  there is also a consistent hypothesis in \(\C(\Phi^*, \A^*, k, \ell)\),
  and \cref{lem: existence of consistent local concept} ensures that
  the algorithm then returns a hypothesis.

  It remains to show that the algorithm satisfies the running time requirements
  while only using local access to the structure \(\A^*\).
  For all \(\bar{a} \in A^k\) and \(\bar{v}^* \in A^\ell\),
  we can bound the size of their neighbourhood by
  \(\abs{\neighb{r}{\A^*}{\bar{a}\bar{v}^*}} \leq (k + \ell) \cdot \sum_{i=0}^r d^i \leq (k + \ell) \cdot (1+d^{r+1})\).
  Therefore, the representation size of the substructure $\Neighb{r}{\A^*}{\bar{a}\bar{v}^*}$ is in
  \(\bigO\big((k + \ell) \cdot d^{r+1} \cdot \log n\big)\).
  Thus, the consistency check in lines 4--8 runs in time
  \(f_{\Phi^*}(\A^*) \cdot t \cdot \bigO\big((k + \ell) \cdot d^{r+1} \cdot \log n\big)\).
  The algorithm checks up to
  \(\abs{N}^\ell \cdot \abs{\Phi^*} \in \bigO\big((tkd^{(2r+1) \ell + 1})^\ell \cdot \abs{\Phi^*}\big)\) hypotheses
  with \(N = \neighb{(2r+1)\ell}{\A^*}{T}\).
  All in all, since \(k\), \(\ell\), \(r\) are considered constant,
  the running time of the algorithm is in
  \(f_{\Phi^*}(\A^*) \cdot (\log n + d + t)^{\bigO(1)}\)
  and it only uses local access to the structure \(\A^*\).
\end{proof}

\subsection{Agnostic PAC Learning with Precomputation}%
\label{sec: agnostic PAC learning with precomputation}

Section~\ref{sec: agnostic PAC learning with precomputation}
is devoted to the proof of Theorem~\ref{thm: agnostic PAC learning with precomputation}.

To obtain a hypothesis that generalises well,
we follow the \emph{Empirical Risk Minimization} rule (ERM)%
~\cite{Shalev-Shwartz:2014:UML:2621980,Vapnik_ERM},
i.e.\ our algorithm should return a hypothesis \(h\)
that minimises the \emph{training error}
\[\text{err}_T(h) \ \coloneqq \ \textstyle\frac{1}{\abs{T}} \cdot \abs{\{(\bar{a},b) \in T : h(\bar{a}) \neq b\}}\]
on the training sequence \(T\).
To prove \cref{thm: agnostic PAC learning with precomputation},
we use the following result from~\cite{Shalev-Shwartz:2014:UML:2621980}.

\begin{lem}[Uniform Convergence~\cite{Shalev-Shwartz:2014:UML:2621980}]%
  \label{lem: uniform convergence}
  Let \(\mathcal{H}\) be a finite class of hypotheses
  \(h \colon A^k \to \{0,1\}\).
  Then \(\mathcal{H}\) is agnostically PAC-learnable
  using an ERM algorithm and
  \[t_\mathcal{H} (\epsilon, \delta) \ \coloneqq \ \left\lceil \frac{2\log(2\abs{\mathcal{H}}/\delta)}{\epsilon^2} \right\rceil.\]
\end{lem}

\begin{proof}[Proof of \cref{thm: agnostic PAC learning with precomputation}]
  We show that the algorithm depicted on the right-hand side of \cref{fig: agnostic PAC learning with precomputation}
  fulfils the requirements from \cref{thm: agnostic PAC learning with precomputation}.
  The algorithm goes through all tuples
  \(\bar{v}^* \in (\neighb{(2r+1)\ell}{\A^*}{T})^\ell\)
  and all formulas \(\phi^*(\bar{x}, \bar{y}) \in \Phi^*\)
  and counts the number of errors that
  \(\llbracket\phi^*(\bar{x}, \bar{y})\rrbracket^{\A^*}(\bar{x}, \bar{v}^*)\)
  makes on \(T\).
  Then it returns the hypothesis with the minimal training error.

  Since \(\Phi^*\) and \(A^\ell\) are finite,
  \(\mathcal{H} = \C(\Phi^*, \A^*, k, \ell)\)
  is finite.
  Thus, using \cref{lem: uniform convergence},
  \(\mathcal{H}\) is agnostically PAC-learnable
  with
 \ $
t_\mathcal{H} (\epsilon, \delta) \ = \ \left\lceil
    \frac{2\log(2\abs{\mathcal{H}}/\delta)}{\epsilon^2} \right\rceil \
  \leq \ \left\lceil \frac{4 \ell \log(\abs{\Phi^*})
      \log(n/\delta)}{\epsilon^2} \right\rceil.
$
  The running time analysis works as in the proof of \cref{thm: exact learning}.
  The algorithm returns a hypothesis in time
  \(f_{\Phi^*}(\A^*) \cdot (\log n + d + t)^{\bigO(1)}\).
  For a training sequence of length \(t = t_\mathcal{H}(\epsilon, \delta)\),
  we obtain a running time in \(f_{\Phi^*}(\A^*) \cdot \big(\log n + d + \log (1/\delta) + 1/\epsilon\big)^{\bigO(1)}\).
\end{proof}

\section{Putting Things Together}\label{sec:conclusion}
Let the collections $\Ps$ and $\SC$ be finite (but $\SC$ may contain
infinite rings or abelian groups), fix
a \emph{finite} set  $\mathscr{S}$ of elements $s\in S\in\SC$, fix
a $q\in\NN$, and
let $\Phi\deff\Phi_{q,k+\ell,\mathscr{S}}$ be the set of
$\WAunPS$-formulas defined in \cref{remark:propertyForLearning}. Let
$\Phi^*$, $\sigma^*$, and $\A^*$ (for all $(\sigma,\Weights)$-structures $\A$)
be as described in
\cref{remark:propertyForLearning}.
By \cref{thm:AlgorithmicDecomp}, $\A^*$ can be computed from
$\A$ in
time $|A|{\cdot} d^{\bigOh(1)}$, where $d$ is the degree of $\A$.
By \cref{remark:propertyForLearning}, the formulas in $\Phi^*$ are $r$-local for a fixed
number $r$, and this implies that model checking for a formula in
$\Phi^*$ on $\A^*$ can be done in time polynomial in $d$.
Combining this with Theorems~\ref{thm: exact learning} and
\ref{thm: agnostic PAC learning with precomputation} yields the
following\footnote{All mentioned algorithms are assumed to have $\PP$- and $\SC$-oracles, so that
  operations $\plusS,\malS$ for $S\in\SC$ and checking if a tuple
  is in $\sem{\P}$ for $\P\in\PP$ takes time $\bigOh(1)$.}.

\begin{thm}\label{thm:LearningFOWA1}
Let $n$ and $d$ denote the size and the degree of $\A$.
\begin{enumerate}[(1)]
\item
  There is an algorithm that solves
  Exact Learning with Precomputation for \(\Phi\)
  and \(\Phi^*\)
  with local access to a
  structure \(\A^*\)
  associated with a structure $\A$
  in time \((\log n + d+t)^{\bigO(1)}\), where \(t\) is the number of training examples.
\item
  There is an \(s \in \NN\) such that,
  given local access to a structure \(\A^*\)
  associated with a structure \(\A\),
  the hypothesis class \(\mathcal{H} \coloneqq \C(\Phi^*, \A^*, k, \ell)\)
  is agnostically PAC-learnable
  with \(t_\mathcal{H} (\epsilon, \delta) = s \cdot \left\lceil \frac{\log(n/\delta)}{\epsilon^2} \right\rceil\)
  via an algorithm that,
  given \(t_\mathcal{H} (\epsilon, \delta)\) examples,
  returns a hypothesis of the form \((\phi^*, \bar{v}^*)\)
  with \(\phi^* \in \Phi^*\) and \(\bar{v}^* \in A^\ell\)
  in time \(\big(\log n + d+ \frac{1}{\epsilon} + \log \frac{1}{\delta}\big)^{\bigO(1)}\) with
  only local access to \(\A^*\).
\end{enumerate}
  Additionally, the algorithms can be chosen such that
  the returned hypotheses can be evaluated
  in time \((\log n + d)^{\bigO(1)}\).
\end{thm}

We conclude with an example that illustrates an application scenario for \cref{thm:LearningFOWA1}.

\begin{exmp}
  Recall the \((\sigma, \Weights)\)-structure \(\A\)
  for the online marketplace
  from part \eqref{exmp:marketplace} of Examples~\ref{example:WStr},
  \ref{example:intuition}, and \ref{example:moreintuition}.
  Retailers can pay the marketplace to advertise their products to consumers.
  Since the marketplace demands a fee for every single view of the advertisement,
  retailers want the marketplace to only show the advertisement
  to those consumers that are likely to buy the product.
  One possible way to choose suitable consumers is to consider only those
  who buy a variety of products from the same or a similar product group
  as the advertised product and who are thus more likely to try
  new products that are similar to the advertised one.
  At the same time, the money spent by the chosen consumers on the product group
  should be above average.

  In the previous examples, we have already seen a formula
  \(\phi_{\textup{spending}}(c)\) that defines consumers
  who have spent at least as much as the average consumer on the product group.
  The formula depends on a formula \(\phi_{\textup{group}}(p)\)
  that defines a certain group of products based on the structure
  of their transactions.
  Due to the connection between graph neural networks
  and the Weisfeiler-Leman algorithm described in \cite{MorrisRitzert_WL},
  we may assume that there is a formula in \(\FO[\sigma]\)
  that at least roughly approximates such a product group.
  Likewise, we might assume that there is a formula
  \(\phi_{\textup{variety}}(c)\) in \(\FO[\sigma]\)
  that defines consumers with a wide variety of products
  bought from a specific product group.
  However, it is a non-trivial task to design such formulas by hand.
  It is even not clear whether there exist better rules
  for finding suitable consumers.
  Meanwhile, we can easily show the advertisement to consumers
  and then check whether they buy the product.
  Thus, we can generate a list with positive and negative examples
  of consumers.
  Since the proposed rule can be defined in \(\WAunPS\)
  as \(\phi_{\textup{advertise}}(c) \deff (\phi_{\textup{variety}}(c) \land \phi_{\textup{spending}}(c))\),
  we can use one of the learning algorithms from
  \cref{thm:LearningFOWA1}
  to find good definitions for \(\phi_{\textup{variety}}(c)\)
  and \(\phi_{\textup{group}}(p)\)
  or to learn an even better definition for
  \(\phi_{\textup{advertise}}(c)\)
  in \(\WAunPS\) from examples.
\end{exmp}

We believe that our results can
be generalised to an extension of $\WAun$ where constructions of
the form $\P(t_1,\ldots,t_m)$ are not restricted to the case that
$|V|=1$ for $V\deff\free(t_1)\cup\cdots\cup\free(t_m)$, but may
also be used in a guarded setting of the form
$\big(\P(t_1,\ldots,t_m)\und\Und_{v,w\in V}\dist(v,w)\,{\leq}\,r\big)$.
It would also be interesting to study non-Boolean classification problems, where classifiers are
described by $\SC$-terms defined in a suitable fragment of $\WA$.
We plan to do this in future work.
\section*{Acknowledgements}
We thank Martin Grohe and Sandra Kiefer for helpful discussions on the subject.

\bibliography{ms}

\end{document}